\documentclass[aps,epsfig,floatfix]{revtex4-1}

\usepackage[latin1]{inputenc}

\usepackage{color,graphicx}
\usepackage{subfigure}
\usepackage{rotating}
\usepackage{amsmath,amsfonts,amssymb,amsthm}

\usepackage[usenames,dvipsnames]{pstricks}
\usepackage{epsfig}
\usepackage{pst-grad}
\usepackage{pst-plot}
\usepackage[final]{pdfpages}

\usepackage[normalem]{ulem}

\newtheorem{theorem}{Theorem}
\newtheorem{proposition}[theorem]{Proposition}
\newtheorem{lemma}[theorem]{Lemma}
\newtheorem{corollary}[theorem]{Corollary}
\newtheorem{definition}{Definition}
\newtheorem*{remark}{Remark}

\def\RR{\mathbb{R}}
\def\CC{\mathbb{C}}

\def\H{\mathcal{H}}

\def\Sym{\textbf{Sym}}
\def\diag{\textbf{diag}}

\def\Sfive{\mathbf{S}_5}
\def\Ssix{\mathbf{S}_6}
\def\Sfour{\mathbf{S}_4}
\def\Brho{\boldsymbol\rho}
\def\Orho{{\mathbf{\mathcal{O}}}_{\Brho}}

\def\Rtt{\mathbf{R_{2\times3}}}

\DeclareMathOperator{\Tr}{Tr}

\begin{document}

\title{On the spectral dependence of separable and classical correlations in small quantum systems}

\author{Gary McConnell and David Jennings}

\affiliation{Controlled Quantum Dynamics Theory Group, Level 12, EEE,\\ Imperial College London, London SW7 2AZ, United Kingdom}

\email{g.mcconnell@imperial.ac.uk}

\date{\today}

\begin{abstract}
We study the correlation structure of separable and classical states in $2\times2$-~and~$2\times3$-dimensional quantum systems with fixed spectra. Even for such simple systems the maximal correlation - as measured by mutual information - over the set of unitarily accessible separable states is highly non-trivial to compute; however for the $2\times2$ case a particular class of spectra admits full analysis and allows us to contrast classical states with more general separable states. We analyse a particular entropic binary relation on the set of spectra and prove for the qubit-qutrit case that this relation alone picks out a unique classical maximum state for mutual information. Moreover the $2\times3$ case is the largest system with such a property.
\end{abstract}

\maketitle

\tableofcontents

\section{Background: the unitary orbit of a bipartite quantum state}

Let $A$ and $B$ be two quantum systems with states of $A$ represented in $m$-dimensional Hilbert space $\H_A=\CC^m$ and those of $B$ in $\H_B=\CC^n$. We assume that $m,n\geq2$.
Let $\rho=\rho_{AB}$ be any state of the joint system $\H_A\otimes\H_B$, with (real) eigenvalues $\lambda_1, \lambda_2, \ldots, \lambda_{mn}$ summing to 1. Quantum measurements performed on the local subsystems will in general reveal correlations between the two states $\rho_A=\Tr_B\rho_{AB}$ and $\rho_B=\Tr_A\rho_{AB}$, some of which may be attributable to entanglement but others of which could be recreated classically in some sense by preparing the joint system in a probabilistic mixture of known product states - that is to say, in a \emph{separable} state.
Furthermore there is a very small discrete subset of these separable states known as the \emph{classical} states: representable by diagonal matrices in the joint computational basis. The question of the extent to which any correlations are ``genuinely quantum'' is key to the resource-based theories of quantum information and quantum computation currently being developed. Indeed there are also many questions in thermodynamics (see for example \cite{santerdav}, \cite{santerdavPRL} and the references contained therein) which arise from viewing correlation as a resource in nanotechnological applications, where correlations between local states in quantum superpositions are demonstrably more powerful than those for classical states.

There is a compelling question in the middle however - what about the correlations of separable states, which are generally non-classical but also not entangled?  To make this distinction, we may speak of \emph{separable correlations} giving the level of correlation inside joint states which are convex sums of product states (but which will not in general be classical), which usually display a higher degree of correlation than the purely classical states associated to the same spectrum. Similarly, \emph{classical correlations} will refer to correlations within classical states. 

Given a fixed state $\rho$ as above, the set of quantum states with identical spectrum are precisely those obtained from $\rho$ via unitary transformations.
Indeed, we may act upon our state $\rho$ via transformations from the unitary group $\mathcal{U}(mn)$ of degree $mn$, generating the \it unitary orbit \rm $\Orho$ containing all quantum states which are \it reversibly \rm obtainable from our starting state $\rho$. The key thing to note is that in traversing a generic unitary orbit we pass through points with only classical or separable correlations and then through a much larger set of inseparable quantum states whose entanglement is linked to the possibility of yet higher correlations. In other words, even though the spectrum remains constant, we nevertheless create and destroy correlations in the course of traversing the orbit $\Orho$: a statement which highlights the dependence of the notions we are discussing, upon the basis in which we have chosen to represent the states.

So within $\Orho$ we have a natural hierarchy of states with the potential for non-zero correlations, within which we would expect the classical states to be the ``lowest'' in some sense, and the pure entangled states to be the ``highest''. 

\begin{figure}[h!btp]
\begin{center}
\mbox{
\subfigure{\includegraphics[width=6in,height=6in,keepaspectratio]{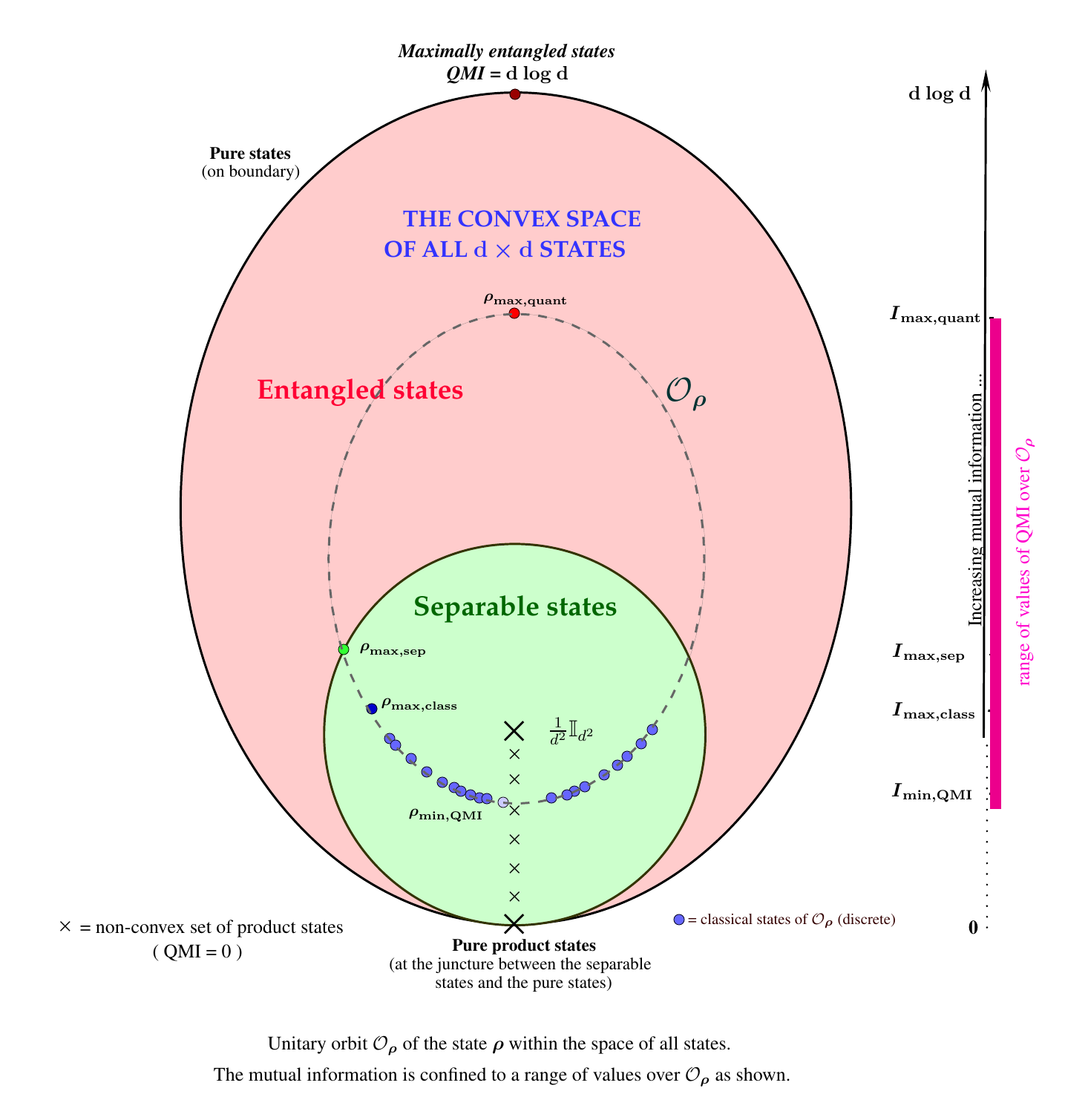}}

}
\end{center}\caption{}\label{egg}
\end{figure}

This setup is depicted schematically in figure~\ref{egg}. The large oval represents the convex set of all states; the inner circle at the bottom is the convex set of separable states, and the entangled states make up the remainder, as in a Venn diagram. Within this the orbit $\Orho$ of a single state is depicted as the boundary of an oval - note that it is not in itself a convex set; however the convex hull of $\Orho$ consists precisely of all of the unitary orbits of spectra which are \it majorised \rm by the spectrum of $\rho$: in particular the maximally mixed state $\frac{1}{d^2}\mathbb{I}_{d^2}$ (which is a unitary orbit consisting of just one point). Quantities which we shall refer to below are noted in the diagram: in particular we must point out that the mutual information (see below) scale on the right is very much a schematic one ... it must be read only in the context of a particular unitary orbit $\Orho$ as shown - otherwise for example the maximally mixed state at the barycentre of the space would always lie ``above'' the minimally correlated state, which is clearly not the case in general.

What is an appropriate measure for this scale? The standard both in quantum and classical information theory is \it mutual information\rm, which loosely speaking measures the distance between a joint state and the product of its reduced subsystems. For any joint system as above we define the quantum mutual information (QMI) to be 
$$I(\rho_{AB}) = S(\rho_A)+S(\rho_B)-S(\rho_{AB}),$$
where for any state $\sigma$ with spectrum $\Lambda_\sigma=\{l_1,\ldots,l_N\}$ we denote by $S(\sigma)=\Tr\left(-\sigma\log\sigma\right)$ its von Neumann entropy. Observe that for a \emph{classical} state the matrix of $\sigma$ will be diagonal in the computational basis and so the definition of von Neumann entropy reduces to the Shannon entropy 
$$H\left((l_1,\ldots,l_N)\right)=\sum_{i=1}^{N}-l_i\log l_i$$
of the probability vector $(l_1,l_2,\ldots,l_N)$. Indeed the definition of QMI then reduces to that of classical mutual information (CMI). However mutual information is not perfect for our purposes, because it does not in any way distinguish between classical and quantum correlations. Indeed it is quite common to have an entangled state with lower mutual information than a classical state: hence the current attempts to define measures which separate quantum correlations from classical ones, such as quantum ``discord'' and quantum ``dissonance''. 
So to get a handle on what sorts of tradeoffs can occur between these quantum and classical correlations, a good starting point is to be able to delimit the maxima and minima of mutual information for each class (i.e.~classical, separable or entangled) of states within a particular orbit. We should also highlight recent work by Partovi~\cite{partovi}, in which he develops a neat, general framework capturing the notion of ``disorder'' in terms of majorisation theory - hence a stronger classification than is provided by entropic measures for example, yielding fewer relations. This approach is in a sense an orthogonal one to ours in that, instead of working with a fixed spectrum and moving over the unitary orbit as we have done below, the \emph{marginal spectra} are fixed, with the total spectrum allowed to vary, revealing what constraints that places upon the possible \emph{minimally disordered} states - be they classical, separable or entangled.

It turns out~\cite{santerdavPRL} that the minimal mutual information within $\Orho$ can always be realised on a classical state, hence \it a fortiori \rm a separable state. Since the classical states form a discrete set it is in principle a straightforward problem to find the minimum (in general it will occur on one of a relatively small suite of permutations identifiable by a simple test - see~\cite{santerdav} - but outside the case $(m,n)=(2,2)$ there is no particular configuration which will be the minimum in all cases). Another way to view this is to remark that being quite a coarse measure, the mutual information is unable to tell us anything about the nature of the state - i.e.~whether it be classical, separable or entangled - for any joint states with sufficiently low correlations.

\begin{remark}
It is important to note that the states we refer to as ``unique'' are only unique up to local unitary operations (and if $m=n$ also the ``transpose'' operation of swapping the two systems $A$ and $B$). Hence we shall tend to refer to unique \emph{classes} of states rather than unique states.
\end{remark}

So we then ask about the maxima. When $m=n$ the maximum overall mutual information occurs for a maximally entangled state~\cite{santerdavPRL} - indeed there will in general be an infinite number of such maximally entangled states yielding the maximal QMI. The cases where $m\neq n$ are messier but the answer is similar.
As with the minimum, a set of conditions can be laid down such that the maximum CMI must occur for a class of states which is one of a relatively small set of candidates; however outside the cases $(m,n)=(2,2)$ and $(2,3)$ this maximum configuration is non-unique. Indeed numerical studies indicate that for $(2,4)$ there are 2 classes which both occur as maxima for different spectra, for $(2,5)$ there are six and for $(3,3)$ there are 18. For comparison we mention that in the case of the minimal CMI, in $(2,2)$ the class is unique (as is the maximum), in $(2,3)$ there are exactly 5 possibilities, and then numerical results indicate that in $(2,4)$ there are 14, in $(2,5)$ there are 42 and in $(3,3)$ there are 18.

Little is known about the \it maximum separable state \rm along a general unitary orbit. The only known way to access it is via convex optimisation using the Peres-Horodecki \it positive partial transpose \rm criterion~\cite{bengtsson}, which outside of the cases $(m,n)=(2,2)$ and $(2,3)$ is only a necessary condition for separability and so does not really assure us of a result anyway. In section~\ref{2bit} we first look at a special class of the $(2,2)$ states where we \it can \rm pin down the maximal separable state, and do some calculations to illustrate its behaviour as contrasted with the maximally and minimally correlated classical states on the same orbit. This is only  achievable because of the neat framework laid out by R.~and~M.~Horodecki~\cite{horodeckis} for understanding the unitary orbit of a two-qubit state. In the $(2,3)$ case we do not have such a framework, and it is consequently much more difficult to understand the big picture.

Section~\ref{3bit} contains the main result of this paper (theorem~\ref{biggun}) where we show that in the case where $(m,n)=(2,3)$, the maximal CMI occurs always (uniquely, up to an action by 12 CMI-invariant transformations) at the state represented by a diagonal matrix containing a fixed ordering of the eigenvalues $\lambda_1,\ldots,\lambda_6$. Curiously this fixed ordering is the same for \emph{every} spectrum, \emph{irrespective of the relative sizes of the eigenvalues}. As we mentioned above, whereas this is also true for $(m,n)=(2,2)$, it is not true for larger joint systems like $(2,4)$ or $(3,3)$.

\newpage
\section{Separable versus classical correlations for the two qubit case}\label{2bit}

We restrict for a moment to the case $(m,n)=(2,2)$. As we shall see, non-trivial features arise even in the simplest possible setting.

So we have a joint system $\H_A\otimes\H_B$ of two qubits in a given state $\rho=\rho_{AB}$ with spectrum $\{a,b,c,d\}$ satisfying $a,b,c,d\geq0$ and $a+b+c+d=1$. There is a representation~\cite{horodeckis} of states $\rho_{AB}$ of such a system in terms of the Pauli matrices, which gives two local reduced Bloch vectors $\mathbf{r}_A,\ \mathbf{r}_B$ at $A$ and $B$, together with a 3-by-3 real  ``correlation matrix'' $\mathbf{T}=(t_{ij})$, giving a total of $3+3+9=15$ real variables parametrising exactly the action of $SU(4)$ on $\rho_{AB}$. Furthermore if we restrict to what they refer to in~\cite{horodeckis} as the \it T-states\rm, namely those states with maximally mixed reductions (hence trivial Bloch vectors but maximal contributions each of $\log2$ to the mutual information) at $A$ and $B$, then by local changes of basis we may arrange that $\mathbf{T}$ is in fact diagonal and so we are reduced to looking in these specific instances at just three real variables $t_{11},\ t_{22}$ and $t_{33}$. Now a natural choice of spanning set for the T-states is the standard Bell basis 
$$|\Phi^+\rangle=\frac{1}{\sqrt{2}}(|00\rangle+|11\rangle),\ |\Phi^-\rangle=\frac{1}{\sqrt{2}}(|00\rangle-|11\rangle),\ |\Psi^+\rangle=\frac{1}{\sqrt{2}}(|01\rangle+|10\rangle)\text{\rm\ and\ }|\Psi^-\rangle=\frac{1}{\sqrt{2}}(|01\rangle-|10\rangle).$$ 
Then from the constraints that $\Tr\rho_{AB}=1$ and that $\rho_{AB}$ be a positive matrix we obtain a tetrahedron $\mathcal{T}$ of $\mathbf{T}$-states with vertices $|\Phi^+\rangle,|\Phi^-\rangle,|\Psi^+\rangle,|\Psi^-\rangle$.

These diagonal matrices may be represented by what they call a \bf t\rm-vector $\mathbf{t}=(t_{11},\ t_{22},\ t_{33})$: the Bell basis elements correspond respectively to the \bf t\rm-vectors $(1,-1,1),\ (-1,1,1),\ (1,1,-1)$ and $(-1,-1,-1)$. In this framework there is a natural way to choose a maximal QMI state~\cite{santerdavPRL} for the given spectrum $\{a,b,c,d\}$ of $\rho_{AB}$: namely, it is the state 
$$\rho_\text{\rm max,QMI} = a|\Phi^+\rangle\langle\Phi^+|+b|\Phi^-\rangle\langle\Phi^-|+c|\Psi^+\rangle\langle\Psi^+|+d|\Psi^-\rangle\langle\Psi^-|,$$
whose \bf t\rm-vector is 
$$\mathbf{t}_\text{\rm max,QMI} = (a-b+c-d,\ -a+b+c-d,\ a+b-c-d).$$
The QMI of this state is
\begin{equation}\label{maxQMI2x2}
I_\text{\rm max,QMI}(a,b,c) = 2\log2-H((a,b,c,d)),
\end{equation}
which is maximal over $\Orho$.
Note also that since $d=1-a-b-c$ we shall view all of these quantities as functions on $\RR^3$ rather than $\RR^4$. The translation from the representation of the states with maximally mixed reductions in the $T$-state picture, back to the eigenvalue-picture is as follows: given a $T$-state vector $\mathbf{t}=(u,v,w)$ with zero local Bloch vectors the corresponding spectrum is
$$(a,b,c) = (\frac{1+u-v+w}{4},\ \frac{1-u+v+w}{4},\ \frac{1+u+v-w}{4})\ .$$
The reason that the T-state setup is so useful for our purposes is that the \it separable \rm states with maximally mixed reduced states and diagonal $\mathbf{T}$-matrix, turn out to be exactly those states whose eigenvalues are all less than or equal to $\frac{1}{2}$. These states trace out an octahedron $\mathfrak{O}$ inside $\mathcal{T}$ which is given in the $T$-coordinate system by $\mathfrak{O}=\mathcal{T}\cap-\mathcal{T}$. Its vertices are $(\pm1,\ 0,\ 0),\ (0,\ \pm1,\ 0),\ (0,\ 0,\ \pm1)$.

Since we are using the joint computational basis and writing things in terms of Pauli matrices, any \it classical \rm state on the orbit $\Orho$ may be written
$$\rho_\text{\rm class}^\tau = \tau(a)|00\rangle\langle00|+\tau(b)|01\rangle\langle01|+\tau(c)|10\rangle\langle10|+\tau(d)|11\rangle\langle11|$$
for some $\tau\in\Sfour$, the symmetric group on four letters. (Here we have arbitrarily allocated the identity element of $\Sfour$ to the state where the eigenvalues $a,b,c,d$ are arranged in alphabetical order down the diagonal which we denote by $\diag(a,b,c,d)$). The local Bloch vectors $\mathbf{r}_A,\ \mathbf{r}_B$ of $\rho_\text{\rm class}^\tau$ are no longer zero in general but rather 
$$\mathbf{r}_A^\tau = (0,\ 0,\ \frac{a+b-c-d}{4}),\ \ \mathbf{r}_B^\tau=(0,\ 0,\ \frac{a-b+c-d}{4}),\ \ \text{\rm with\ $\mathbf{t}$-vector\ }\mathbf{t}^\tau=(0,\ 0,\ \frac{a-b-c+d}{4}).$$

We assume from now on that $a\geq b\geq c\geq d\geq 0$. 
We know from~\cite{santerdavPRL} (or see appendix~\ref{2by2}) that under these conditions the state
\begin{eqnarray*}
\rho_\text{\rm min} & = & \diag(a,b,c,d)
\end{eqnarray*}
will give us the minimal QMI on $\Orho$:
\begin{equation}\label{minQMI2x2}
I_\text{\rm min}(a,b,c) = h(a+b)+h(a+c)-H((a,b,c,d)),
\end{equation}
and that the state
\begin{eqnarray*}
\rho_\text{\rm max,class} & = & \diag(a,d,c,b)
\end{eqnarray*}
corresponding to the permutation $\tau=(2,4)$ will give the maximal CMI on $\Orho\ $: 
\begin{equation}\label{maxCMI2x2}
I_\text{\rm max,class}(a,b,c) = h(a+c)+h(b+c)-H((a,b,c,d)).
\end{equation}
We have used the standard convention in~(\ref{minQMI2x2})~and~(\ref{maxCMI2x2}) that $h$ is the \it binary entropy function \rm 
$$h(x)=-x\log x-(1-x)\log(1-x).$$

Finally we define $I_\text{\rm max,sep}$ to be the maximal QMI attainable on a \bf separable \rm state $\rho_\text{\rm max,sep}$ in the orbit $\Orho$. Note that in general these maximal and minimal states will not be unique; whereas the value of the information can be abstractly uniquely defined.

We wish to analyse the behaviour of the functions~$I_\text{\rm max,QMI}$,~$I_\text{\rm max,sep}$,~$I_\text{\rm min}$~and~$I_\text{\rm max,class}$ as we roam over $\Orho$ for some fixed spectrum $\{a,b,c,d\}$. Whereas $\rho_\text{\rm max,sep}$ is difficult to find for a generic state, for illustrative purposes, we may restrict to the subset of spectra for which $\rho_\text{\rm max,QMI}$ lies inside the octahedron $\mathfrak{O}$ of separable states, for then we are guaranteed that $\rho_\text{\rm max,sep}$ will coincide with $\rho_\text{\rm max,QMI}$, namely those with eigenvalues all less than or equal to $\frac{1}{2}$. Thus by construction,
\begin{equation}\label{entsep}
I_\text{\rm max,QMI} = I_\text{\rm max,sep}
\end{equation}
for all of the states we shall be considering in this section.
Define ``gap'' functions $\gamma_\text{\rm max}$ and $\gamma_\text{\rm min}$ as the differences between the quantity in~(\ref{entsep}), and those in~(\ref{minQMI2x2})~and~(\ref{maxCMI2x2}) respectively:
$$\gamma_\text{\rm max}(a,b,c)=2\log2-h(a+c)-h(b+c)$$
and
$$\gamma_\text{\rm min}(a,b,c)=2\log2-h(a+b)-h(a+c).$$
These represent the gaps in mutual information as we travel over different spectra, between the maximal QMI states and their maximal and minimal counterparts in the classical subset.
Indeed $\gamma_\text{\rm max}$ is a signature function for the non-classicality of the state space: the states we are considering are not entangled; nevertheless they are able to manifest greater mutual information than would a purely classical state with the same spectrum.

For the avoidance of confusion we should point out that by definition, $\gamma_\text{\rm max}\leq\gamma_\text{\rm min}$.

The functions $\gamma_\text{\rm max}$ and $\gamma_\text{\rm min}$ are defined on the domain of spectra:
$$\mathcal{D}^\text{\rm id}=\{(a,b,c)\in\RR^3\boldsymbol{:}\ \frac{1}{2}\geq a\geq b\geq c\geq (1-a-b-c)\geq0\},$$
which is a kind of pyramid with an irregular quadrilateral base. Its five vertices are at the points $V_1 = (\frac{1}{2},\frac{1}{2},0)\ \ $ (which is the apex of the pyramid), $V_2 = (\frac{1}{4},\frac{1}{4},\frac{1}{4})$, $V_3 = (\frac{1}{3},\frac{1}{3},\frac{1}{3})$, $V_4 = (\frac{1}{2},\frac{1}{4},\frac{1}{4})$ and $V_5 = (\frac{1}{2},\frac{1}{6},\frac{1}{6})$.

Note that for each rearrangement $\tau\in\Sfour$ of the positions of the eigenvalues we obtain another domain $\mathcal{D}^\tau$: in total these 24 domains glue together to form an octahedron which is a linear image of the regular octahedron $\mathfrak{O}$ (see the diagrams below). As we cross from one fundamental domain into another the functions $\gamma_\text{\rm max}$ and $\gamma_\text{\rm min}$ will need to be re-defined in order to take into account the new ordering of the eigenvalues. 

We restrict our attention therefore to the behaviour of $\gamma_\text{\rm max},\ \gamma_\text{\rm min}$ on the convex region $\mathcal{D}^\text{\rm id}$. Now $-H$ is a convex function on its domain the unit interval $[0,1]$ and since the maps from $\RR^3$ to $\RR$ given by $(a,b,c)\mapsto (a+b)$, $(a,b,c)\mapsto (a+c)$ and $(a,b,c)\mapsto (b+c)$ are all linear it follows that 
$\gamma_\text{\rm max}$ and $\gamma_\text{\rm min}$ are also convex on $\mathcal{D}^\text{\rm id}$. (Note that these ``gap'' functions will in fact be convex on the whole octahedral domain; however one needs always to rearrange the arguments in the definitions as remarked above). Hence $\gamma_\text{\rm max}$ and $\gamma_\text{\rm min}$ will attain their maximal values on an extremal point of $\mathcal{D}^\text{\rm id}$, which means one or more of the vertices $V_1,V_2,V_3,V_4,V_5$ above. By direct calculation we find that the maximum of $\gamma_\text{\rm max}$ is $\frac{3}{4}\log3-\log2$ (around 0.1308 in the natural logarithm) and it occurs at the point~$V_4$. That is to say, this is the largest possible deviation of mutual information (from the classical values) once one is allowed the full scope of the quantum state space for these particular spectra.

For $\gamma_\text{\rm min}$ the maximum is $\log2$ and it occurs at the point~$V_1$. (Hence upon acting by $\Sfour$ we see that the maximal points for $\gamma_\text{\rm min}$ are actually the six vertices of the octahedron).

We can show directly from the definitions that the \it minimal \rm values of these functions are always zero: 
$$\gamma_\text{\rm max}(a,b,c)=0\text{\rm\ if and only if }a=b=\frac{1}{2}-c=\frac{1}{2}-d;$$
and
$$\gamma_\text{\rm min}(a,b,c)=0\text{\rm\ if and only if }a=b=c=d=\frac{1}{4}\ .$$
Hence $\gamma_\text{\rm max}$ is zero on the line joining $V_1$ and $V_2$; while $\gamma_\text{\rm min}$ is zero only at the point $V_2$ (which represents the maximally mixed state).

On the next few pages we include some depictions of the behaviour of these two functions $\gamma_\text{\rm max},\ \gamma_\text{\rm min}$ on the whole octahedron $\mathfrak{O}$, using the translation above from $(a,b,c)$-space to the $T$-state space. The first picture shows the splitting of the octahedron into the (image of the) fundamental regions $\mathcal{D}^\tau$. Notice that the diagrams are all in ``\bf t\rm-vector space'' and so care should be taken when thinking about probability distributions in terms of the spectra $\{a,b,c,d\}$ to use formulae like those given in the first part of this section in order to pass from the spectrum to the octahedron and vice-versa.

It is worth making a few comments on these results. Firstly, the set of spectra, and their corresponding unitary orbits in state space, do \emph{not} coincide with those unitary orbits lying entirely within the set of separable states. Such orbits correspond to the \emph{absolutely separable states} \cite{KusZyczk} - namely those quantum states for which it is impossible to unitarily generate entanglement. Indeed, it has been shown \cite{Ishizaka, Verstraete} that absolutely separable states have spectra that obey $a \le c+2 \sqrt{bd}$, and so are found to be a proper subset of the spectra that we consider. The implication of this is that orbits exist that contain entangled states, but attain their maximal mutual information on separable states. This can be seen more explicitly by considering the so-called \emph{maximally entangled mixed states} (MEMS), being those states for which it is impossible to unitarily increase a given measure of entanglement.

For the case of two qubits, the MEMS have been found to take the form
\begin{eqnarray}\label{MEMS}
\rho_{\mbox{\tiny MEMS}} &=& a |\psi^-\rangle \langle \psi^-| + b |00\rangle \langle 00| + c |\psi^+ \rangle \langle \psi^+| + d| 11\rangle  \langle 11|,
\end{eqnarray} 
modulo local unitaries, and have maximal concurrence $C_{\mbox{\tiny MEMS}} = \mbox{max}(0,a - c - 2\sqrt{bd})$. Indeed, for a fixed spectrum, the state $\rho_{\mbox{\tiny MEMS}}$ not only maximizes concurrence, but also maximizes the negativity, the relative entropy of entanglement and the entanglement of formation \cite{Verstraete}. It is immediately clear from (\ref{MEMS}) that while $\rho_{\mbox{\tiny MEMS}}$ might maximize entanglement, it generally does \emph{not} have maximally mixed marginals, and cannot be a maximum of the QMI, which is somewhat surprising. Therefore, in the generic case of mixed quantum states there exist competing mechanisms between quantum and classical correlations over the orbit of the state, which stands in contrast with the case of pure quantum states, for which all correlation measures (quantum, classical and total) are simultaneously maximized on the maximally entangled states.

\begin{figure}[h!btp]\label{octa}
\begin{center}
\mbox{
\subfigure{\includegraphics[width=4in,height=4in,keepaspectratio]{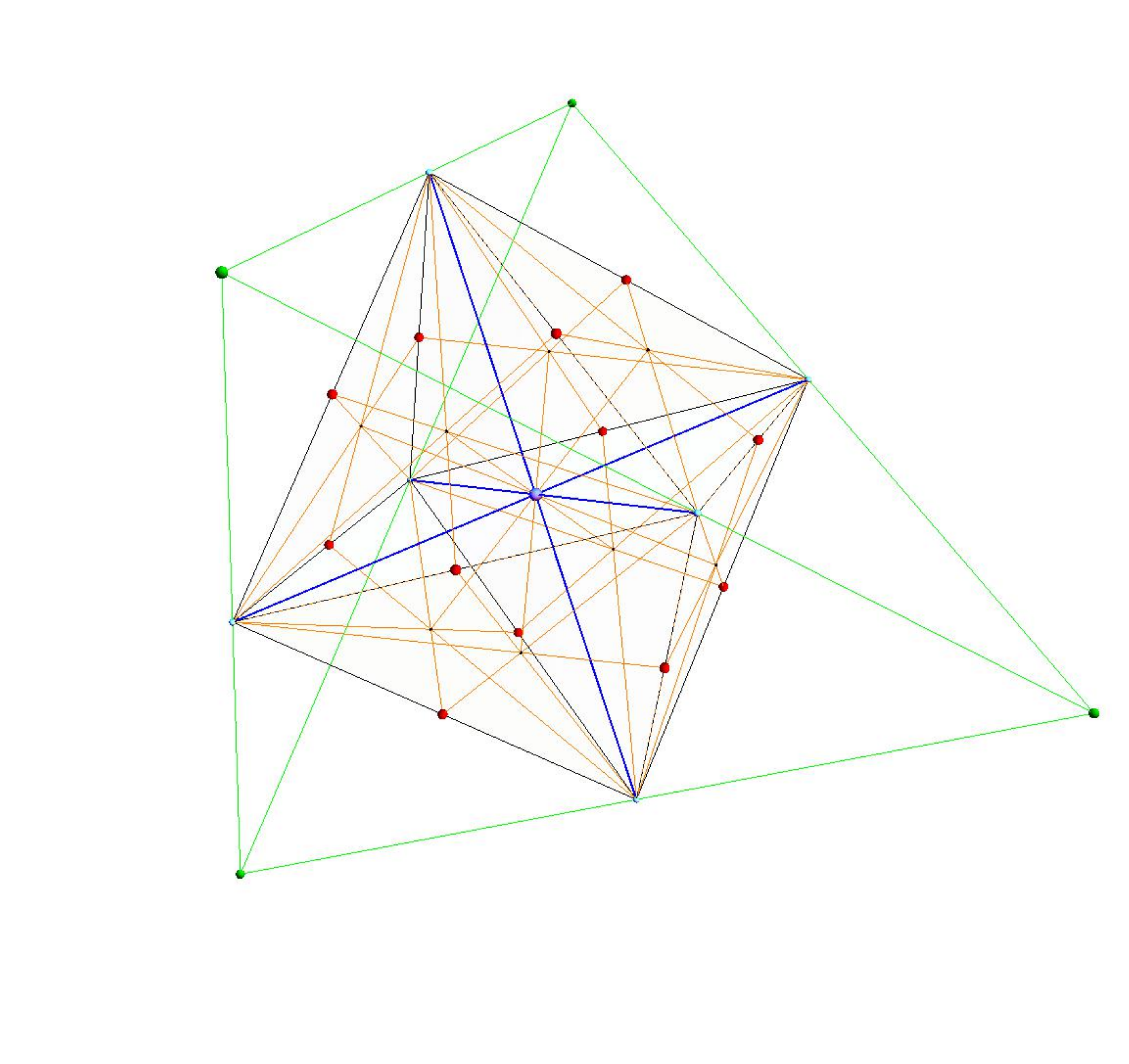}}
}
\caption{The octahedron $\mathfrak{O}$ constituted of the 24 fundamental regions $\cal{D}^\tau$, all sitting inside the green tetrahedron of $T$-states. The red dots are the maxima for $\gamma_\text{\rm max}$; the blue lines are the minima. The central blue dot is the unique maximally mixed state.}
\end{center}
\end{figure}

\begin{figure}[h!btp]\label{fundreg}
\begin{center}
\mbox{
\subfigure{\includegraphics[width=4in,height=3.6in,keepaspectratio]{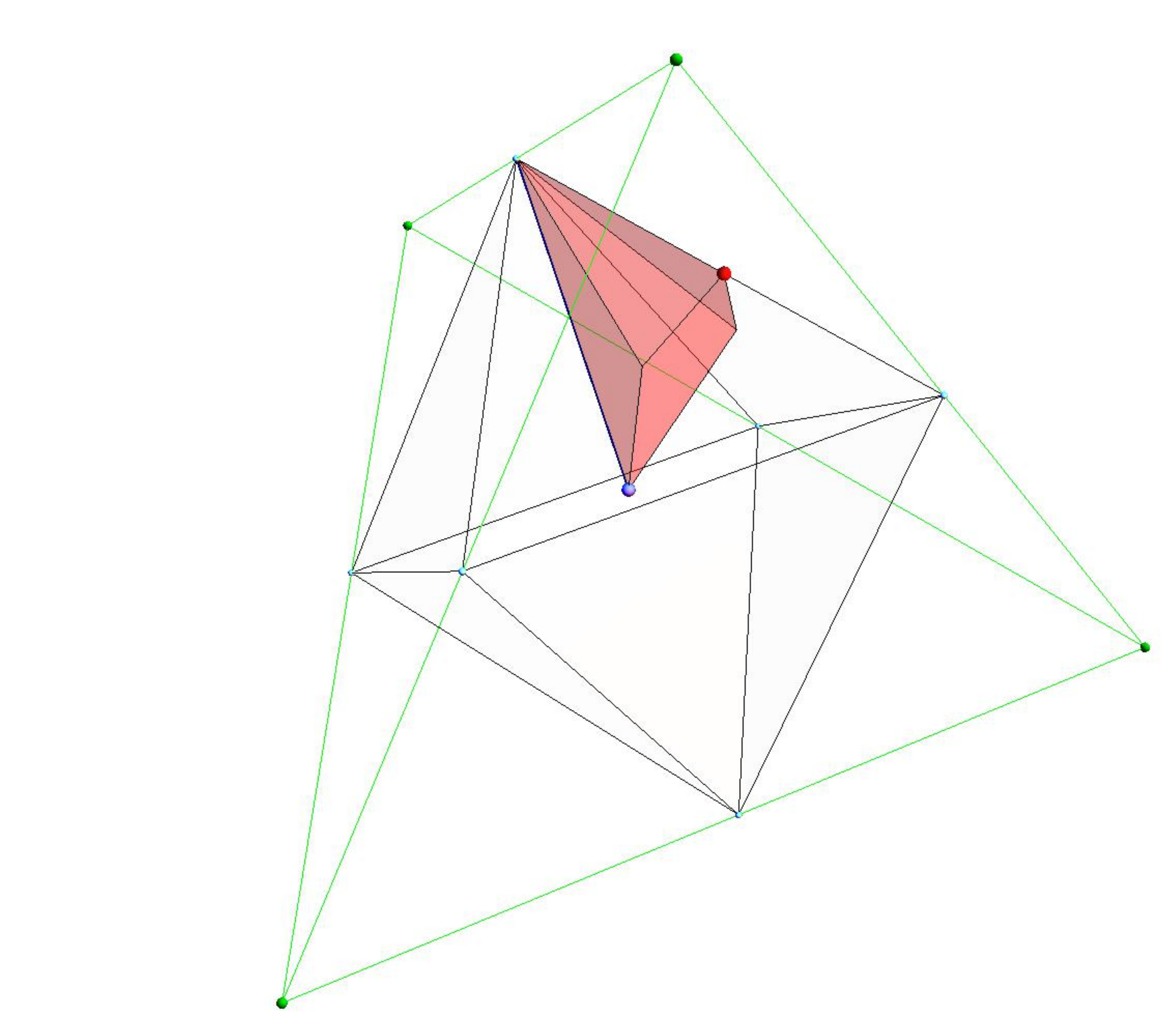}}
}
\caption{A special case of the first diagram: one fundamental region $\cal{D}^\text{\rm id}$ embedded in the octahedron $\mathfrak{O}$, all embedded inside the green tetrahedron of $T$-states whose green dot vertices are the Bell states. The vertices of the octahedron are the classically correlated states like $\frac{1}{2}(|00\rangle\langle00|+|11\rangle\langle11|)$ which all have CMI $=\log2$.}
\end{center}
\end{figure}

\begin{figure}[h!btp]\label{gammax}
\begin{center}
\mbox{
\subfigure{\includegraphics[width=4in,height=3.6in,keepaspectratio]{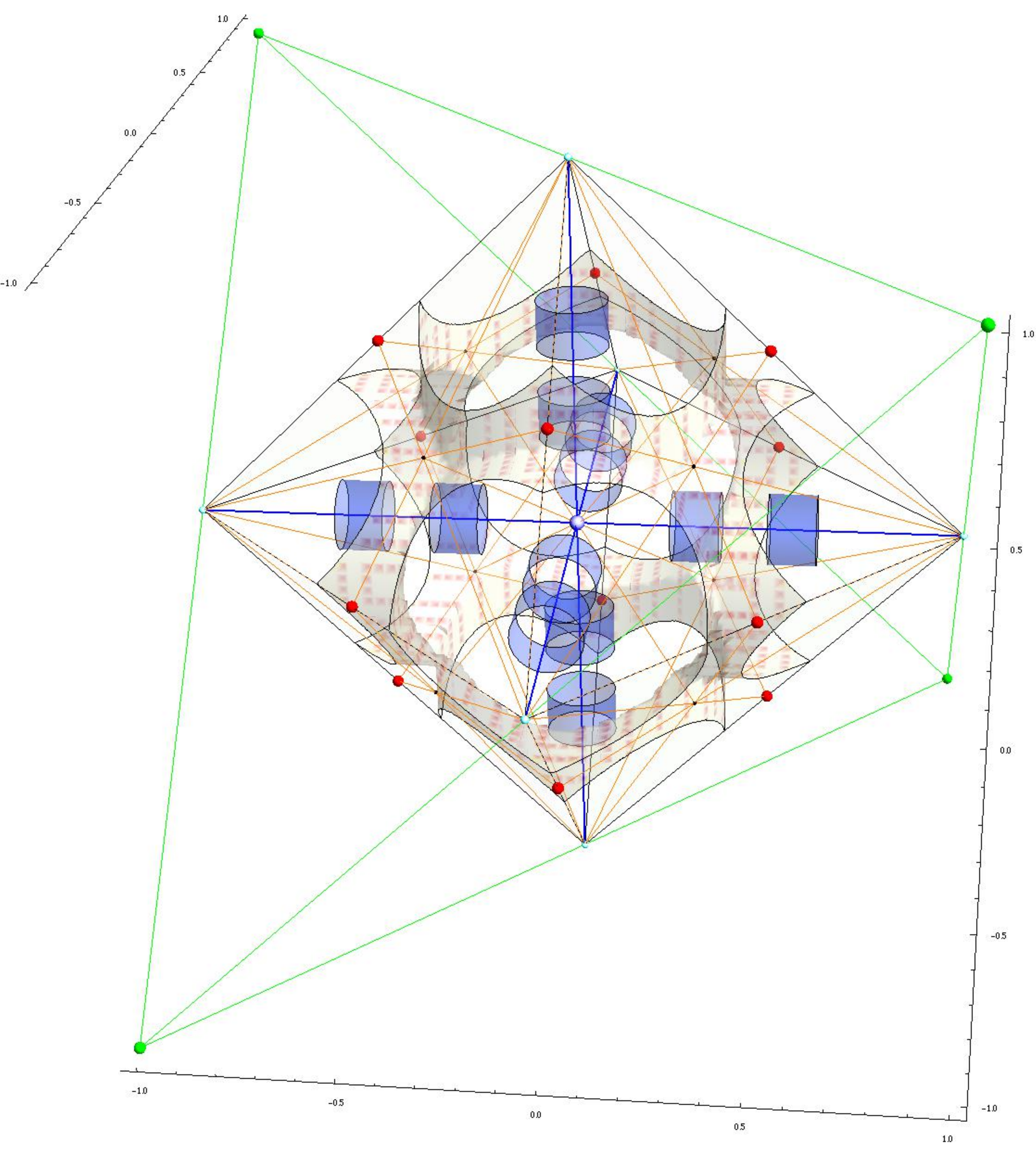}}
}
\caption{The octahedron $\mathfrak{O}$ showing some sections of contours of the function $\gamma_\text{\rm max}$ together with the fundamental regions. $\gamma_\text{\rm max}$ measures the gap between maximal separable and maximal classical correlations for a fixed spectrum; each spectrum is here represented by a single point in $\mathfrak{O}$.}
\end{center}
\end{figure}

\begin{figure}[h!btp]\label{gammax_tempmap}
\begin{center}
\mbox{
\subfigure{\includegraphics[width=4in,height=3.6in,keepaspectratio]{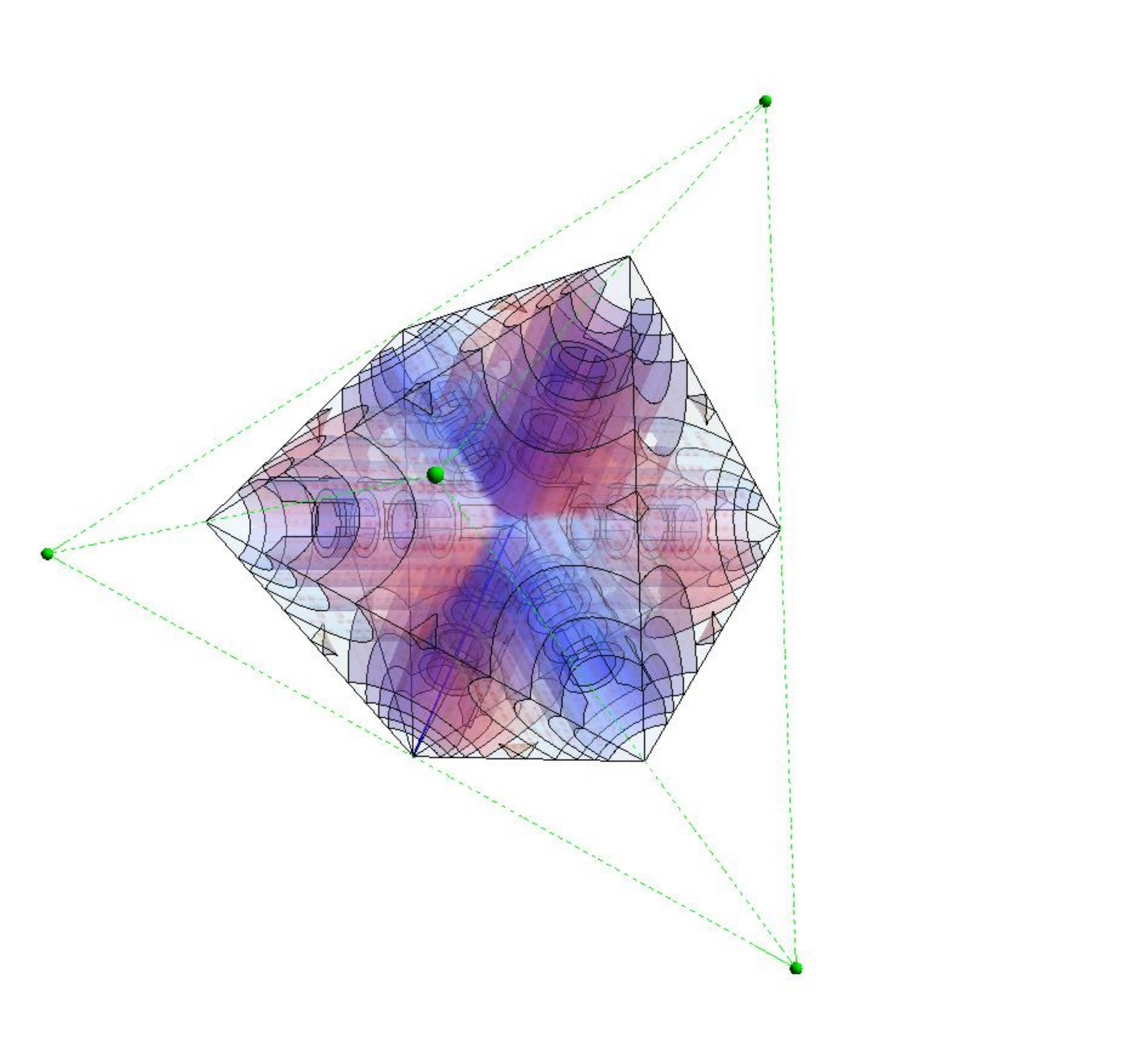}}
}
\caption{Alternative view of the function $\gamma_\text{\rm max}$ with more contours. The colour scheme in both these diagrams goes from blue (low) to red (high): the maximal points are the midpoints of the octahedron's 12 edges; whereas the centre of the cylindrical regions will be the blue lines shown in the fundamental regions above on which $\gamma_\text{\rm max}=0$}
\end{center}
\end{figure}

\begin{figure}[h!btp]\label{gammin}
\begin{center}
\mbox{
\subfigure{\includegraphics[width=4in,height=3.6in,keepaspectratio]{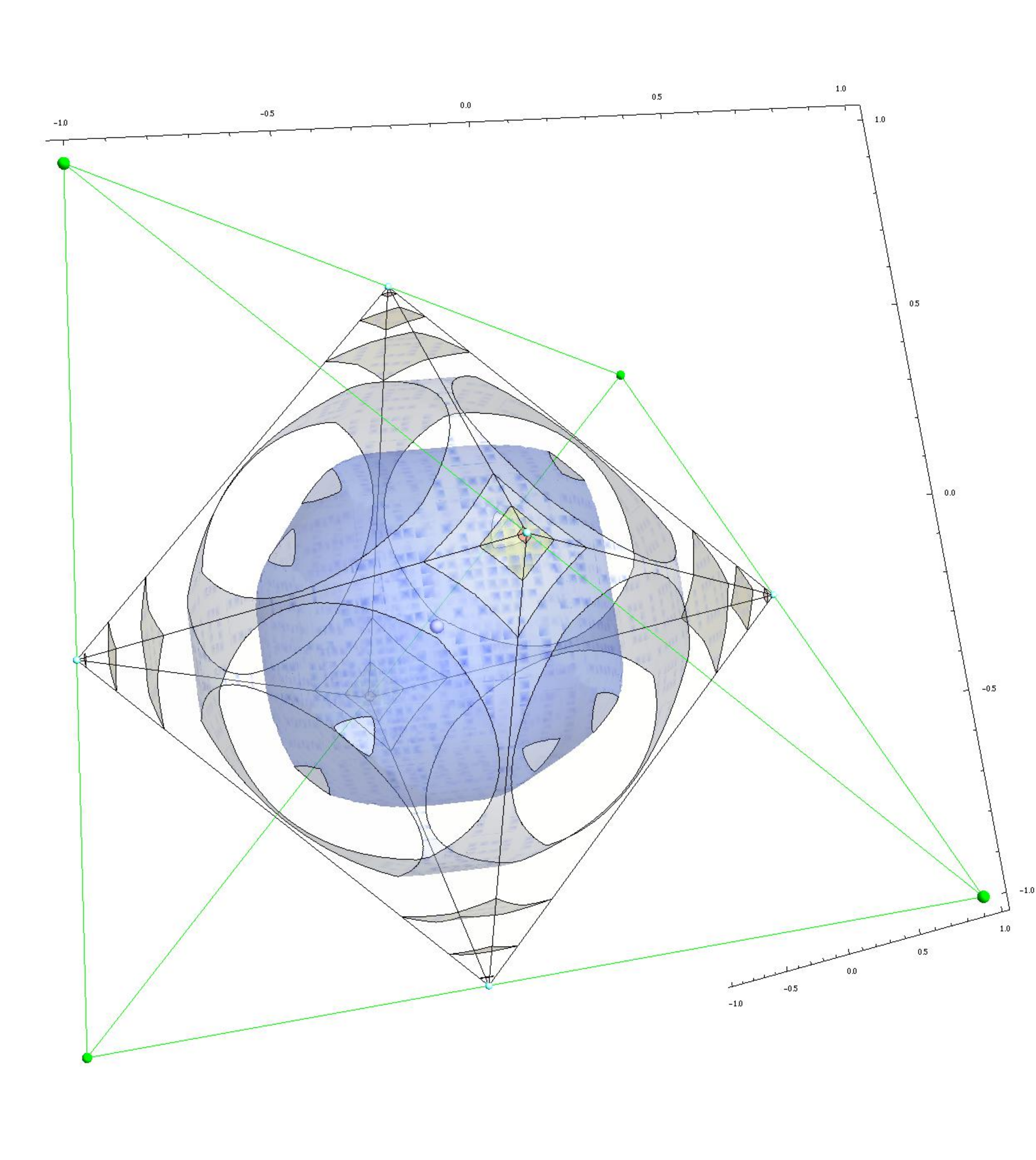}}
}
\caption{The octahedron $\mathfrak{O}$ showing the contours of the function $\gamma_\text{\rm min}$, which measures the full range of separable correlations attainable over a fixed spectrum. Broadly speaking the function increases with distance from the minimum value of zero at the barycentre of the octahedron: its extrema occur at the vertices of $\mathfrak{O}$. The intermediate contours are a kind of truncated cube or cuboctahedron}
\end{center}
\end{figure}

\newpage

\section{An entropic binary relation for $2\times3$ systems and the unique maximal classical state}\label{3bit}

So we have seen that a $2\times2$ system admits a straightforward analysis of the ``classical gap'' $I_\text{\rm max,class}-I_\text{\rm min,class}$ and that for many spectra (namely those where the eigenvalues are all $\leq\frac{1}{2}$), the ``separable-classical gap'' $\gamma_\text{\rm max}$ is also fairly easy to calculate. However already for the $2\times3$ case it becomes relatively non-trivial even to determine the classical gap. Indeed we make this the focus of this last section and discover a curious property of this $2\times3$ world: namely that $I_\text{\rm max,class}$ is determined solely by an ordering of the eigenvalues and not by their relative sizes. It seems that there is just enough information to pin down the maximum (but again rather surprisingly \emph{not} the minimum~\cite{santerdav}); but increasing either dimension renders this impossible.

We revert for a moment to the general setting of the introduction, in order to fix some ideas.
So let $A$ and $B$ be two quantum systems with states of $A$ represented in $m$-dimensional Hilbert space $\H_A=\CC^m$ and those of $B$ in $\H_B=\CC^n$.
Let $\rho_{AB}$ be any state of the joint system $\H_A\otimes\H_B$, with (real) eigenvalues $\lambda_1, \lambda_2, \ldots, \lambda_{mn}$. We may consider the \emph{classical} state lying in the unitary orbit of $\rho_{AB}$, which may be viewed purely as a diagonal matrix of probabilities summing to~$1$:
\begin{equation}\label{classdiag}
\diag(
\lambda_1, \lambda_2 , \ldots , \lambda_{mn} ).
\end{equation}

As we observed above, by implementing unitaries whose effect is simply to send classical states to classical states - that is, to permute the eigenvalues - we arrive at a series of different possibilities for the subsystems $\rho_A=\Tr_B\rho_{AB}$ and $\rho_B=\Tr_A\rho_{AB}$ obtained by taking the respective partial traces of the joint system. In this setting QMI reduces to CMI: a function which is purely defined in terms of partial sums of the eigenvalues $\{\lambda_k\}$.

Now suppose that we are given an ordering on the set of eigenvalues $\lambda_1>\lambda_2>\ldots>\lambda_{mn}$, but no further information on their relative sizes.
We may then ask the question: is there a particular arrangement of these which will guarantee \it a priori \rm to yield the minimal or maximal values of CMI associated with this entire class of diagonal matrices?

As we mentioned in the first section, the answer in the simplest interesting case $(m,n)=(2,2)$ is that the maximum and the minimum may both be found by considerations of majorisation~\cite{santerdav}, as there are only 3 distinct equivalence classes of matrices under the CMI map. In the next simplest case $(m,n)=(2,3)$ the maximum is determined \it a priori \rm and further there are $5$ ``minimal'' matrices~\cite{santerdavPRL}, one of which will be the minimum in any given instance (and indeed all of which \bf do \rm occur in specific examples, meaning that \it a priori \rm the set of minima cannot be whittled down any further without more stipulations on the relative sizes of eigenvalues). Not surprisingly, beyond these low-dimensional instances nothing terribly definitive can be said because the relative gaps between successive eigenvalues come to play too great a role. Indeed, the real surprise is that a definite maximum occurs in the $2\times 3$-case. This means that the maximal CMI has an easy \it a priori \rm determination in the one-qubit-by-one-qutrit context. 
The remainder of this paper is concerned with exploring the structure of this maximal CMI and establishing this unique maximum configuration.

So let $m=2$, $n=3$.
Throughout this paper when speaking about the $2\times3$-case we shall fix our set of six eigenvalues of the quantum state $\rho=\rho_{AB}$ as $\{a,\ b,\ c,\ d,\ e,\ f\}$ with~$a+b+c+d+e+f=1$ and assume that $a>b>c>d>e>f>0$ 
(we shall usually treat these as though they were \bf strict \rm inequalities in order to derive sharper statements but everything is valid if we allow $\geq$ instead).
The main result is as follows.

\begin{theorem}\label{biggun}
With notation as above, the permutation $[\ a, d, e, f, c, b\ ]$ giving rise to marginal probability vectors $\big(a+f,\ c+d,\ b+e\big)$ and $\big( a+d+e,\ b+c+f\big)$
has maximal CMI among all $720$ possible permutations of $[\ a, b, c, d, e, f \ ]$. 

\bf This is the case \emph{irrespective of the sizes of the gaps between} $\mathbf {a,b,c,d,e,f}$.\rm
\end{theorem}

In order to prove the theorem we need some preliminary ideas.

\subsection{Definitions: classical mutual information, majorisation and the entropic binary relation $\rhd$}

We set up the framework of the problem for general $m,n$.

\subsubsection{The classical mutual information (CMI) attached to an $m\times n$ probability matrix}

Suppose we are given a matrix $\rho_{AB}$ in the form~(\ref{classdiag}) with a given splitting into a pair of subsystems $A$ of dimension $m$ and $B$ of dimension $n$, so that we may arrange the eigenvalues $\lambda_k$ in an $m\times n$-matrix as follows:
\begin{equation}\label{rowsncols}
\bordermatrix{
&c_1 & c_2 & \ldots & c_n\cr
r_1&\lambda_1 & \lambda_2 & \ldots & \lambda_n \cr
r_2&\lambda_{n+1} & \lambda_{n+2} & \ldots & \lambda_{2n} \cr
\vdots&\vdots & \vdots & \ddots & \vdots \cr
r_m&\lambda_{(m-1)n+1} & \lambda_{(m-1)n+2} & \ldots & \lambda_{mn} \cr
}=P.
\end{equation}
As shown we let the row sums be denoted by $r_i=\sum_{j=1}^n\lambda_{(i-1)n+j}$ for $i=1,\ldots,m$ and similarly for the column sums: $c_j=\sum_{i=1}^m\lambda_{(i-1)n+j}$ for $j=1,\ldots,n$.
Then by the definition of the partial trace map (equivalently, the contraction of a tensor along a particular index) we see that the density matrices $\rho_A$ and $\rho_B$ referred to above are now the diagonal matrices $\diag(r_1,\ldots,r_m)$ and $\diag(c_1,\ldots,c_n)$ respectively.

So $P$ has the form of a joint probability matrix where the marginal probabilities are given by the $r_i$ and the $c_j$. To define the classical mutual information (see~\cite{cover}, \S2.3) we take the sum of the entropies of the $r_i$ and the $c_j$ over all $i,j$ and then subtract the sum of the individual entropies of the $\lambda_k$, for $k=1,\ldots,mn$. Formally:

\begin{definition}\label{cmile}
With notation as above, the classical mutual information $I(P)$ of the matrix $P$ is given by
\begin{equation}\label{CMIdef}
I(P) = \sum_{i=1}^m-r_i\log r_i + \sum_{j=1}^n-c_j\log c_j - \sum_{k=1}^{mn}-\lambda_k\log\lambda_k.
\end{equation}
We will often write $H(x)=-x\log x$ for $x\in [0,1]$ and so we may rewrite~(\ref{CMIdef}) as
\begin{equation*}
I(P) = \sum_{i=1}^m H(r_i) + \sum_{j=1}^n H(c_j) - \sum_{k=1}^{mn}H(\lambda_k).
\end{equation*}
\end{definition}

\subsubsection{Majorisation between two $m\times n$ probability matrices}

For definitions and basic results connected with majorisation, see~\cite{bhatia} and~\cite{marshall}. We shall use the standard symbol $\succ$ to denote majorisation.
For any $m\times n$-matrix $M$ denote by $\mathbf{r}(M)\in\RR^m$ the vector of marginal probabilities represented by the sums of the rows of $M$ and similarly by $\mathbf{c}(M)\in\RR^n$ the vector of marginal probabilities created from the sums of the columns of $M$.

\begin{lemma}\label{majoris}
Let $M_1,M_2$ be two probability matrices.
If $\mathbf{r}(M_1)\succ\mathbf{r}(M_2)$ and if $\mathbf{c}(M_1)\succ\mathbf{c}(M_2)$, then
$$I(M_1)\leq I(M_2).$$
\end{lemma}

\begin{proof}
See~\cite{santerdav}: it follows from the fact that $H$ is a Schur-concave function (see \cite{bhatia}, \S II.3).
\end{proof}

It should be pointed out that the converse is definitely NOT true: indeed it is this very failure which enables us to prove the main theorem of this paper.

\begin{definition}\label{critta}
If the hypotheses of Lemma \ref{majoris} hold then we write
$$M_1\succ M_2$$
and we shall say that $M_1$ majorises $M_2$: \bf but note that this matrix terminology is not standard.
\end{definition}

By symmetry the relation of majorisation between matrices is invariant under row swaps and/or column swaps. In addition if $m=n$ then the majorisation relation is also invariant under transposition.

\subsubsection{An entropic binary relation $\rhd$ among $m\times n$ probability matrices}

The entropic binary relation $\rhd$, which we now define, is the key to proving theorem~\ref{biggun}.
If we consider the class of $(mn)!$ matrices formed by permuting the entries in the matrix $P$ in~(\ref{rowsncols}) and look at the CMI of each of these, there is a rigid \it a priori \rm partial order which arises between them~\cite{mcmaj}. Most of this can be explained by majorisation considerations; however in low dimensions there is a substantial set of relations which depends on a much finer graining than majorisation gives. This fine-graining is an entropic binary relation which is implied by the stronger relation of majorisation: see proposition~\ref{majmcmaj}.

The general relation is defined as follows. Recall from above the definition of the classical mutual information $I(P)$ of a probability matrix~$P$. For any positive integer $N$ we denote by $\mathbf{S}_N$ the symmetric group on $N$ letters.
\begin{definition}\label{qm}
Let $P=(p_{ij})$ be any $m\times n$ probability matrix and let $Q$ be an $m\times n$ matrix obtained by some permutation of the elements of $P$ \rm(\it that is, viewed as vectors: $Q=P^\sigma$ for some $\sigma\in\mathbf{S}_{mn}$\rm)\it.
Suppose that a complete ordering is given of the $p_{ij}$. We say that $P\rhd Q$ if it can be shown \rm a priori solely using this ordering of the entries\it, that $I(Q)-I(P)$ is non-negative. 

In other words given any ordered probability vector $(p_{ij})$, if we arrange its elements into the orders displayed in~$P$ and~$Q$ then $I(Q)-I(P)\geq0$.

NB: In order to keep the terminology consistent with that of majorisation, we have adopted the convention that $P\rhd Q$ corresponds to $I(P) \leq I(Q)$.
\end{definition}
That is to say, given an \it a priori \rm ordering of the elements of the matrix, such a relation $P\rhd Q$ holds \bf irrespective \rm of the relative sizes of these matrix entries.

\begin{remark}
We mentioned above the connection with the symmetric group $\mathbf{S}_{mn}$. The partial order arising on the matrices gives a partial order on the space of cosets of $\mathbf{S}_{mn}$ modulo a subgroup representing row and column swaps (see for example section~\ref{scub}), because the relations are guaranteed to hold for all $m\times n$ matrices depending as they do only upon the particular arrangement of the $mn$ elements. This points to a deeper connection with combinatorial group theory which we explore in~\cite{mcmaj}.
\end{remark}

In order to see what ~$\rhd$ means in the case which will most interest us - that of a simple transposition - we consider a general $m\times n$ probability matrix $P=(p_{ij})$ with no assumed order among the entries $p_{ij}$. Let $\tau$ be any transposition acting on $P$, interchanging two elements which we shall refer to as $\alpha$ and $\beta$ (by a slight abuse of notation, since the positions and their values will be referred to by the same symbols). The following diagram illustrates this action of $\tau$ on $P$: we write $P^\tau$ for the image of $P$ under $\tau$.

\begin{equation}
\label{matricks}
\bordermatrix{
&&&&c_\beta&&c_\alpha&\cr
&p_{11} & p_{12} & \ldots & \ldots & \ldots & \ldots & p_{1n} \cr
&p_{21} & p_{22} & \ldots & \ldots & \ldots & \ldots & p_{2n} \cr
&\vdots & \vdots & \vdots & \vdots & \vdots & \vdots & \vdots \cr
r_\alpha&\ldots & \ldots & \ldots & \ldots & \ldots & \alpha & \ldots \cr
&\vdots & \vdots & \vdots & \vdots & \vdots & \vdots & \vdots \cr
r_\beta&\ldots & \ldots & \ldots & \beta  & \ldots & \ldots & \ldots \cr
&\vdots & \vdots & \vdots & \vdots & \vdots & \vdots & \vdots \cr
&p_{m1} & p_{m2} & \ldots & \ldots & \ldots & \ldots & p_{mn}\cr
}=P\ ;\ \ \ \
\bordermatrix{
&&&&c_\beta^\tau&&c_\alpha^\tau&\cr
&p_{11} & p_{12} & \ldots & \ldots & \ldots & \ldots & p_{1n} \cr
&p_{21} & p_{22} & \ldots & \ldots & \ldots & \ldots & p_{2n} \cr
&\vdots & \vdots & \vdots & \vdots & \vdots & \vdots & \vdots \cr
r_\alpha^\tau&\ldots & \ldots & \ldots & \ldots & \ldots & \beta  & \ldots \cr
&\vdots & \vdots & \vdots & \vdots & \vdots & \vdots & \vdots \cr
r_\beta^\tau&\ldots & \ldots & \ldots & \alpha & \ldots & \ldots & \ldots \cr
&\vdots & \vdots & \vdots & \vdots & \vdots & \vdots & \vdots \cr
&p_{m1} & p_{m2} & \ldots & \ldots & \ldots & \ldots & p_{mn}\cr
}=P^\tau.
\end{equation}

Without loss of generality we may stipulate that \bf as matrix entries \rm $\alpha>\beta$ (if they are equal there is nothing to be done). We wish to compare $I(P)$ with $I(P^\tau)$. Note firstly that by the definition of CMI, the difference $I(P^\tau)-I(P)$ depends only on the rows and columns containing $\alpha,\beta$. All of the rest of the terms vanish as they are not affected by the action of $\tau$. We denote by $r_\alpha$ (respectively $r_\beta$) the sum of the entries in the row of $P$ which contains $\alpha$ (respectively $\beta$), and by $c_\alpha$ (respectively $c_\beta$) the sum of the entries in the column of $P$ which contains $\alpha$ (respectively $\beta$). Similarly, we denote by $r_\alpha^\tau,r_\beta^\tau,c_\alpha^\tau,c_\beta^\tau$ the image of these quantities under the action of $\tau$. See the diagram~(\ref{matricks}) above.

NB: $r_\alpha^\tau,c_\alpha^\tau$ (respectively, $r_\beta^\tau,c_\beta^\tau$) no longer contain $\alpha$ (respectively~$\beta$), but rather~$\beta$ (respectively~$\alpha$).

So the quantity we are interested in becomes
\begin{equation}
I(P^\tau)-I(P) = H(r_\alpha^\tau)-H(r_\alpha)+H(r_\beta^\tau)-H(r_\beta)+H(c_\alpha^\tau)-H(c_\alpha)+H(c_\beta^\tau)-H(c_\beta),
\end{equation}
with the proviso that if $\alpha$ and $\beta$ happen to be in the same row (respectively column) then the $r_\bullet^\circ$ (respectively, $c_\bullet^\circ$) terms vanish.
The terms on the right hand side are grouped in pairs of the form $\pm(H(x+(\alpha-\beta))-H(x))$, which means we may write it in a more suggestive form:
\begin{equation}\label{slump}
I(P^\tau)-I(P) = (\alpha-\beta)\left(-\frac{H(r_\alpha)-H(r_\alpha^\tau)}{\alpha-\beta}+\frac{H(r_\beta^\tau)-H(r_\beta)}{\alpha-\beta} -\frac{H(c_\alpha)-H(c_\alpha^\tau)}{\alpha-\beta}+\frac{H(c_\beta^\tau)-H(c_\beta)}{\alpha-\beta}\right).
\end{equation}
In order to use calculus we need the machinery of Lagrangian means (see chapter VI \S2.2 of \cite{bullen}).

\begin{definition}
Let $\varphi$ be a continuously differentiable and strictly convex or strictly concave function defined on a real interval $I$, with first derivative $\varphi'$. Define the {\bf Lagrangian mean $\mu_\varphi$ associated with $\varphi$} \it to be:
\begin{equation}\label{lagrean}
\mu_\varphi(a,b)  = \begin{cases}  {\varphi'}^{-1}\left(\frac{\varphi(b)-\varphi(a)}{b-a}\right)&\mbox{if }b\neq a \\
                                    a &\mbox{if }b=a                 \end{cases}
\end{equation}
for any $a,b\in I$, where ${\varphi'}^{-1}$ denotes the unique (on $I$, by virtue of strict convexity/concavity and differentiability) inverse of $\varphi'$.
\end{definition}
In other words, $\mu_\varphi$ is the function which arises from the Lagrangian mean value theorem in the process of going from the points $(a,\varphi(a))$ and $(b,\varphi(b))$ subtending a secant on the curve of $\varphi$, to the unique (in this case) point $\mu_\varphi(a,b)\in[a,b]$ where the slope of the tangent to the curve $\varphi$ is equal to that of the secant. See the diagram below.

\begin{figure}[h!btp]\label{secant}
\begin{center}
\mbox{
\subfigure{\includegraphics[width=3in,height=3in,keepaspectratio]{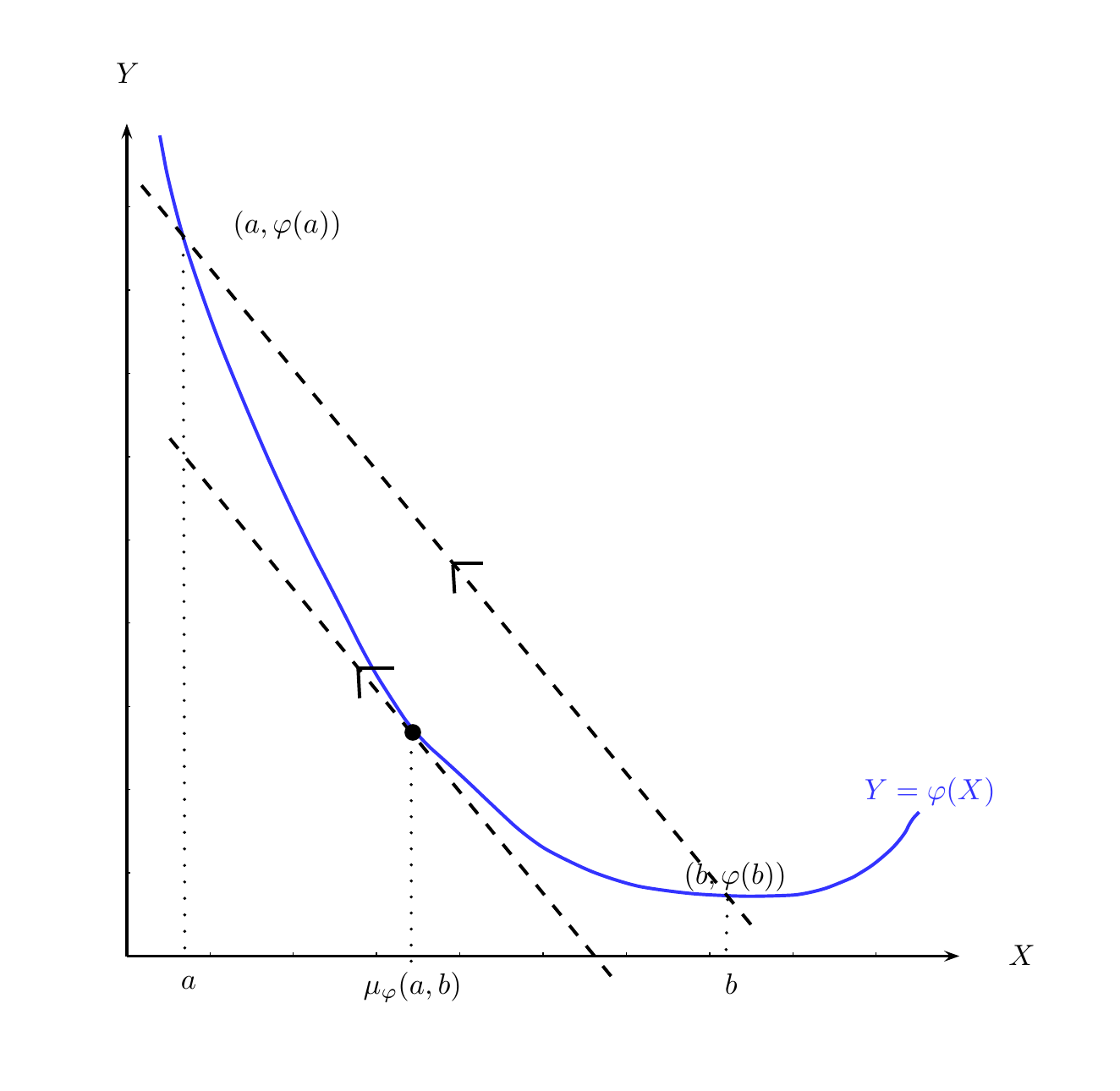}}
}
\caption{Definition of $\mu_\varphi$}
\end{center}
\end{figure}

Each of the arguments for the function $H$ in (\ref{slump}) lies in the interval $I=[0,1]$.  Since $H$ is well-defined and indeed strictly concave and infinitely differentiable on $I$ we may rewrite (\ref{slump}) as:
\begin{eqnarray}
I(P^\tau)-I(P) & = & (\alpha-\beta)\left(-H'(\mu_H(r_\alpha^\tau,r_\alpha))+H'(\mu_H(r_\beta,r_\beta^\tau))
-H'(\mu_H(c_\alpha^\tau,c_\alpha))+H'(\mu_H(c_\beta,c_\beta^\tau))\right)\\
& = & (\alpha-\beta)\log \frac{\mu_H(r_\alpha^\tau,r_\alpha)\mu_H(c_\alpha^\tau,c_\alpha)}
{\mu_H(r_\beta,r_\beta^\tau)\mu_H(c_\beta,c_\beta^\tau)},\label{wang}
\end{eqnarray}
the second line following from the fact that in our context $\varphi'(x) = H'(x) = -(1+\log(x))$.
Since $(\alpha-\beta)>0$ by hypothesis, in order to determine which matrix gives higher CMI we only need consider the relative sizes of the numerator and denominator of the argument of the logarithm. So it is enough to study the quantity
\begin{eqnarray}
\mu_H(r_\alpha^\tau,r_\alpha)\mu_H(c_\alpha^\tau,c_\alpha)-
\mu_H(r_\beta,r_\beta^\tau)\mu_H(c_\beta,c_\beta^\tau).
\label{wing}
\end{eqnarray}

We are now in a position to re-state what is meant by the entropic binary relation~$\rhd$ for this special case of a transposition.

\begin{lemma}
With notation as above, $P\rhd P^\tau$ if and only if it can be shown \rm a priori \it that the quantity in {\rm (\ref{wing})} is non-negative.\qed
\end{lemma}

To study the function $\mu_H$ in more detail we shall need the following lemmas.

\begin{lemma}\label{trance}
Let $u\leq v\leq w\leq z$ be any four positive numbers satisfying $v+w\geq u+z$. Then $vw\geq uz$. 
\end{lemma}

\begin{proof}
Let $\xi=v-u\geq0$ so that $v=u+\xi$ and $z\leq w+\xi$. Then $vw=uw+\xi w\geq uw+\xi u\geq uz$.
\end{proof}

\begin{lemma}\label{crucifix}
Let $\psi$ be a concave monotonically increasing function of the non-negative real numbers taking positive values. Let $p<q<r<s$ be positive real numbers satisfying $q+r\geq p+s$. Then $\psi(q)+\psi(r) \geq \psi(p)+\psi(s)$ and consequently:
\begin{equation}\label{coddy}
\psi(q)\cdot \psi(r) \geq \psi(p)\cdot \psi(s).
\end{equation}
\end{lemma}

\begin{remark}\label{queer}
The condition $q+r \geq p+s$ is sufficient to prove {\rm (\ref{coddy})} but it is not necessary, as the example $\psi={\rm identity}$ and $p=0.1,q=0.4,r=0.4,s=0.9$ shows.
\end{remark}

\begin{proof}
We first show that $$\psi(q)+\psi(r) \geq \psi(p)+\psi(s).$$
Suppose to the contrary that $\psi(p)+\psi(s) > \psi(q)+\psi(r)$: on rearranging we obtain
\begin{equation*}\psi(s)-\psi(r) > \psi(q)-\psi(p).\end{equation*}
By a similar rearrangement of the hypothesis of the lemma we know that $s-r\leq q-p$ and so combining these and using the fact that all terms are positive:
$$\frac{\psi(s)-\psi(r)}{s-r} > \frac{\psi(q)-\psi(p)}{q-p}.$$
By the mean value theorem there exist $\gamma\in[r,s]$ and $\delta\in[p,q]$ such that $\psi'(\gamma) > \psi'(\delta)$. Since $\psi$ is concave it follows that $\psi'$ is monotonically non-increasing, so this implies in turn that $\gamma<\delta$. But the intervals $[p,q]$ and $[r,s]$ are disjoint with the first entirely less than the second, hence $\gamma>\delta$, which is the desired contradiction. 

Now $0\leq\psi(p)\leq\psi(q)\leq\psi(r)\leq\psi(s)$ by the hypothesis that $\psi$ is monotonically increasing, so inequality (\ref{coddy}) follows from lemma~\ref{trance} on setting $u=\psi(p)$, $v=\psi(q)$, w=$\psi(r)$, z=$\psi(s)$.
\end{proof}

We now prove some facts about $\mu_H$ which will give us an insight into the sign of the quantity in~(\ref{wing}).

\begin{lemma}\label{lollipop}
Fix $t\in(0,1)$. For $x\in(0,1-t)$:

(i) $\mu_H(x,x+t)>0$ and is strictly monotonically increasing in $x$;

(ii) $\mu_H(x,x+t)$ is strictly concave in $x$;

(iii) $\frac{1}{e} < \frac{1}{t}(\mu_H(x,x+t)-x) < \frac{1}{2}$.
\end{lemma}

Note that (iii) says that the Lagrangian mean of $x$ and $y$ occurs between $x+\frac{y-x}{e}$ and $x+\frac{y-x}{2}$. Both extremes occur in the limit, so \it a priori \rm we cannot narrow the range down further than this.

\begin{proof}
Solving~(\ref{lagrean}) explicitly for $\varphi=H$ we see that $\mu_H$ is in fact what is known as the {\it identric mean of $x$ and $y$\rm}:
\begin{equation*}
\mu_H(x,y) = e^{-1}\left(\frac{y^y}{x^x}\right)^\frac{1}{y-x},
\end{equation*}
or if we set $t=y-x$:
\begin{eqnarray*}
\mu_H(x,x+t) & = & e^{-1}\left(\frac{(x+t)^{(x+t)}}{x^x}\right)^\frac{1}{t}\\
             & = & e^{-1}(x+t)(1+\frac{t}{x})^\frac{x}{t}\ .
\end{eqnarray*}
From this, the fact that $\mu_H(x,x+t)>0$ for $x,t>0$ may be seen directly. Taking the first derivative with respect to $x$ gives
$$\frac{\partial}{\partial x}\left(\mu_H(x,x+t)\right) = e^{-1}(1+\frac{x}{t})(1+\frac{t}{x})^{\frac{x}{t}}\log(1+\frac{t}{x})$$
which once again is positive for $x,t>0$, proving that indeed $\mu_H(x,x+t)$ is strictly monotonically increasing in $x$ for fixed $t$. This proves (i).

To prove (ii) we take the second derivative with respect to $x$ (writing $\log^2(X)$ for $\left(\log(X)\right)^2$):
$$\frac{\partial^2}{\partial x^2}\left(\mu_H(x,x+t)\right) = e^{-1}(1+\frac{t}{x})^{\frac{x}{t}}\left[\frac{-1}{x}+\frac{t+x}{t^2}\log^2(1+\frac{t}{x})\right].$$
We need to establish that this is always negative: this will be the case if and only if the right-hand term in the square brackets is negative. So we must show that for $x,t>0$:
$$\frac{1}{x} > \frac{t+x}{t^2}\log^2(1+\frac{t}{x}),$$
which for $x>0$ is the same as showing that
$$(\frac{x}{t}+\frac{x^2}{t^2})\log^2(1+\frac{t}{x}) < 1.$$
Since $t>0$ is assumed fixed we may define a new variable $\xi=1+t/x$ and rewrite the left hand side as
$$\frac{\xi}{(\xi-1)^2}\log^2\xi$$
which is $<1$ if and only if its (positive) square root is $<1$, since $\xi$ takes values only between $(1+t)$ and $\infty$. So we need to show that for $\xi>1$,
$$\frac{\sqrt{\xi}\log\xi}{\xi-1}<1$$
or equivalently,
\begin{equation*}
\xi-\sqrt{\xi}\log\xi-1>0.
\end{equation*}
The limit of the left hand side of this expression as $\xi$ tends to $1$ from above is $0$. So it is enough to show that for any $\xi>1$, the derivative:
$$1-\frac{1+\log\sqrt{\xi}}{\sqrt{\xi}},$$
of this left hand side is positive,
which by a change of variable to $z=\sqrt\xi$ (which preserves our domain $\xi>1$) is simply the statement
$$1+\log z < z,$$
which is a standard fact about logarithms (see for example~\cite{hardy} \S 5.2.4).

We only require (i) and (ii) for the proof of theorem~\ref{biggun}, so since (iii) follows by similar techniques we omit the proof.

\end{proof}

We shall need the following sufficient condition for the entropic relation. Consider the four terms which constitute the \it first \rm argument in each of the instances of the function $\mu_H$ in (\ref{wing}), namely
\begin{equation}\label{titanic}
r_\alpha^\tau,c_\alpha^\tau,r_\beta,c_\beta.
\end{equation}
Observe that there are no \it a priori \rm relationships between the sizes of these quantities.
Let us consider the possible orderings of the four terms based upon what we know of the ordering of the matrix elements of $P$. In principle there are $24$ such possibilities; however in certain instances of small dimension such as our $2\times 3$ case, most of these may be eliminated and we are left with only a few orderings.

In looking at~(\ref{matricks}) for the special case where $m=2$ (and $n$ is any integer) one sees that the relations in the following diagram must always hold (ie irrespective of the values of the probabilities), where a downward arrow between values $x_\phi$ and $y_\psi$ indicates that $x_\phi>y_\psi$ \it a priori\rm.

\begin{figure}[h!btp]\label{diamonds}
\begin{center}
\mbox{
\subfigure{\includegraphics[width=6in,height=6in,keepaspectratio]{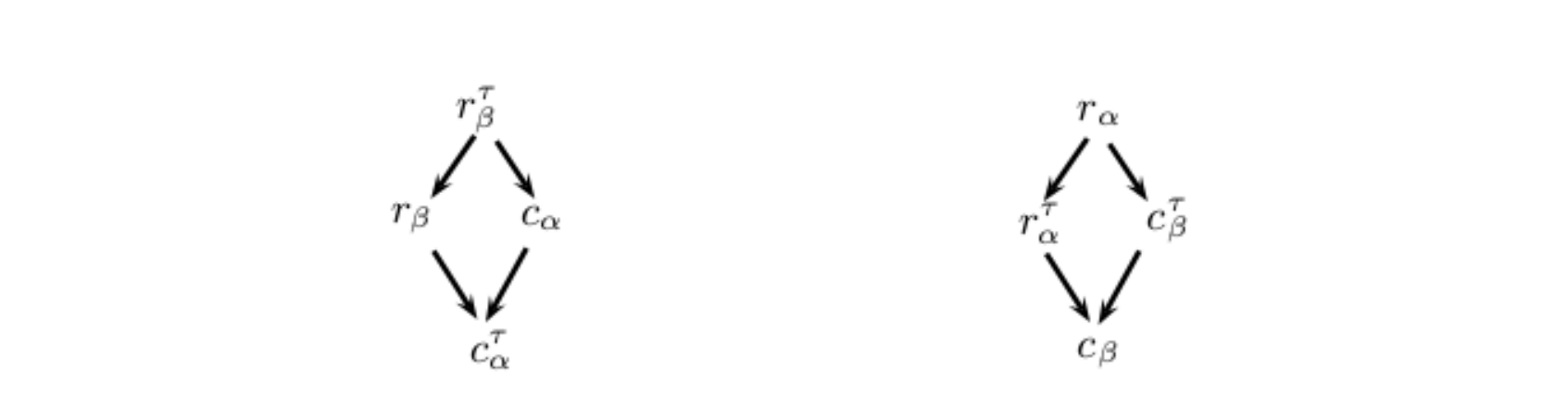}}
}
\end{center}\caption{All fixed relations between the quantities $r_\alpha,\ r_\beta,\ c_\alpha,\ c_\beta,\ r_\alpha^\tau,\ r_\beta^\tau,\ c_\alpha^\tau,\ c_\beta^\tau$ in the $2\times n$ case.}
\end{figure}

\begin{proposition}\label{titrate}
Suppose that the \rm a priori \it minimum element in (\ref{titanic}) is either $r_\beta$ or $c_\beta$. In addition suppose that we can verify \rm a priori \it that
$r_\beta+c_\beta \leq r_\alpha^\tau+c_\alpha^\tau.$
Then $P\rhd P^\tau$.

Conversely, suppose that the \rm a priori \it minimum element in (\ref{titanic}) is either $r_\alpha^\tau$ or $c_\alpha^\tau$ and in addition suppose that we can verify \rm a priori \it that
$r_\beta+c_\beta \geq r_\alpha^\tau+c_\alpha^\tau.$
Then $P\lhd P^\tau$.
\end{proposition}

\begin{proof}
We prove the first assertion; the second follows by symmetry. 

Without loss of generality (since we could at this stage equally consider the transposed matrices) we may assume that the minimum element in (\ref{titanic}) is $c_\beta$. Now if $r_\beta$ is \bf not \rm the maximum element in (\ref{titanic}) then one of $r_\alpha^\tau,c_\alpha^\tau$ is larger than $r_\beta$, hence both $r_\beta$ and $c_\beta$ are dominated by at least one or both of $r_\alpha^\tau,c_\alpha^\tau$, and so by the monotonicity of $\mu_H(x,x+t)$ for fixed $t$ (part (i) of lemma \ref{lollipop}) the expression in (\ref{wing}) must be non-negative, meaning $P\rhd P^\tau$ as required. So suppose to the contrary that $r_\beta$ \bf is \rm the maximum element in (\ref{titanic}), meaning that the ordering of the elements is either
$$c_\beta \leq r_\alpha^\tau \leq c_\alpha^\tau \leq r_\beta$$
or
$$c_\beta \leq c_\alpha^\tau \leq r_\alpha^\tau \leq r_\beta.$$
By rewriting (\ref{wing}) in a more explicit form and writing $t$ for $\alpha-\beta$, we need to show that
\begin{equation}\label{sminch}
\mu_H(r_\alpha^\tau,r_\alpha^\tau+t)\mu_H(c_\alpha^\tau,c_\alpha^\tau+t)-
\mu_H(r_\beta,r_\beta+t)\mu_H(c_\beta,c_\beta+t) \geq 0.
\end{equation}
But using parts (i) and (ii) of lemma~\ref{lollipop} we see that viewed simply as a function of $x$ for fixed values of $t$, $\mu_H(x,x+t)$ satisfies the hypotheses of lemma~\ref{crucifix}.
So finally we let $\psi(x) = \mu_H(x,x+t)$ for fixed $t=\alpha-\beta$, let $p=c_\beta,s=r_\beta$ and set $\{q,r\}=\{r_\alpha^\tau,c_\alpha^\tau\}$ in the appropriate ordering.
Then using the hypothesis of the proposition that $r_\beta+c_\beta \leq r_\alpha^\tau+c_\alpha^\tau$, we obtain the result (\ref{sminch}) from lemma~\ref{crucifix}.
\end{proof}

To tie~$\rhd$ back to majorisation we have the following result.
\begin{proposition}\label{majmcmaj}
Let $P,P^\tau$ be as above. Then
\begin{equation}
\left(P\succ P^\tau\right)  \Rightarrow  \left(P\rhd P^\tau\right).
\end{equation}
Furthermore if $\alpha$ and $\beta$ belong to the same row or column then the two notions of majorisation and~$\rhd$ are the same.
\end{proposition}

\begin{proof}
In essence this is just lemma~\ref{majoris} and definition~\ref{critta}; however because the techniques are needed below, we prove it explicitly.

By a standard result on majorisation (see Corollary II.1.4 on p.31 of \cite{bhatia}) since all terms arising from rows or columns not containing $\alpha$ or $\beta$ are identical for both matrices, $P\succ P^\tau$ may be simplified to the statement that (as vectors):
$$(r_\alpha,r_\beta)\succ(r_\alpha^\tau,r_\beta^\tau){\rm\ and\ }
(c_\alpha,c_\beta)\succ(c_\alpha^\tau,c_\beta^\tau).$$
Now $r_\alpha>r_\alpha^\tau$ by definition, and since conversely $r_\beta<r_\beta^\tau$, it follows that
the elements of each set may be ordered as follows:
\begin{equation}\label{fritter}
r_\alpha>\max\{r_\alpha^\tau,r_\beta^\tau\}>\min\{r_\alpha^\tau,r_\beta^\tau\}>r_\beta,
{\rm\ \ and\ \ }
c_\alpha>\max\{c_\alpha^\tau,c_\beta^\tau\}>\min\{c_\alpha^\tau,c_\beta^\tau\}>c_\beta.
\end{equation}

In the simple case where $\alpha$ and $\beta$ belong to the same row, we may set the row terms in (\ref{wing}) to 1, and it becomes apparent that $P\rhd P^\tau$ is exactly the statement that $\mu_H(c_\alpha^\tau,c_\alpha)>\mu_H(c_\beta,c_\beta^\tau)$, which by part (i) of lemma \ref{lollipop} is the same as saying that $c_\alpha^\tau>c_\beta$ (since the quantity $t = y-x = \alpha-\beta$ is the same for both sides). But $c_\alpha^\tau+c_\beta^\tau=c_\alpha+c_\beta$
and so we must have that $c_\alpha>c_\beta^\tau$. Since we already know (as $\alpha>\beta$) that $c_\alpha>c_\alpha^\tau$ it follows that
$$(c_\alpha,c_\beta)\succ(c_\alpha^\tau,c_\beta^\tau)$$
that is to say $P\succ P^\tau$.
Conversely if $P\succ P^\tau$ then plugging $c_\alpha^\tau>c_\beta$ into (\ref{wing}) - again ignoring the row terms - implies $P\rhd P^\tau$.
So in this case it is clear that $P\succ P^\tau$ is the same thing as $P\rhd P^\tau$. An identical argument works for the other simple case where $\alpha$ and $\beta$ belong to the same column.

Finally let us consider the case where $\alpha$ and $\beta$ are in different rows and columns. Suppose that $P\succ P^\tau$; we must show that $P\rhd P^\tau$. But a look at the relationships between row and column sums in (\ref{fritter}) shows that this is a straightforward application of proposition~\ref{titrate}.
\end{proof}

\subsection{Proof of the main theorem}

So far we have constructed an abstract framework for the study of our entropic binary relation $\rhd$; moreover we have shown that it is a necessary condition for majorisation.
For the rest of the paper we specialise to the case of theorem~\ref{biggun}, namely where $m=2$ and $n=3$ and as always $a>b>c>d>e>f>0$. 
In the terminology of definition~\ref{qm} we need to view the permutation in the statement of theorem~\ref{biggun} as a \it matrix\rm, which we shall call
\begin{equation}\label{exxon}
X = \left( \begin{array}{ccc}a & d & e\\
f & c & b\end{array}\right).
\end{equation}

\subsubsection{The canonical matrix class representatives $\mathbf{R_{2\times3}}$}\label{scub}

Recall our original aim: given an ordering of these numbers $\{a,b,c,d,e,f\}$ we wished to establish whether there was an \it a priori \rm permutation which would give us the minimal and/or the maximal possible mutual information.
There are $6!=720$ possible permutations of these elements, giving a set of matrices which we shall refer to throughout as~$\cal{M}_{\rm 2\times3}$. However since simple row and column swaps do not change the CMI, and since there are $12=|\mathbf{S}_3|.|\mathbf{S}_2|$ such swaps, we are reduced to only $60=720/12$ different possible values for the CMI (provided that $\{a,b,c,d,e,f\}$ are all distinct: clearly repeated values within the elements will give rise to fewer possible CMI values). 

For convenience we shall standardize the form of a set of representatives of these~60 CMI-invariant classes of matrices. This set of chosen representatives will be referred to as~$\mathbf{R_{2\times3}}$.
Since we may always make $a$ the top left-hand entry of any of the matrices in~$\cal{M}_{\rm 2\times3}$ by row and/or column swaps, we set a basic form for our matrices as $M = \left( \begin{array}{ccc}a & x & y\\u & v & w\end{array}\right)$, where (as sets) $\{x,y,u,v,w\}=\{b,c,d,e,f\}$.
This leaves us with only $5!=120$ possibilities which we further divide in half by requiring that $x>y$. So our final form for representative matrices will be:
\begin{equation}
M = \left( \begin{array}{ccc}a & x & y\\
u & v & w\end{array}\right),\text{\rm\ with\ }x>y\text{\rm\ and\ }a=\max\{a,x,y,u,v,w\}.
\label{mates}
\end{equation}

This yields our promised 60 representatives~$\mathbf{R_{2\times3}}$ in the form (\ref{mates}) for the 60 possible CMI values associated with the \bf fixed \rm set of probabilities $\{a,b,c,d,e,f\}$. We now need to further subdivide~$\mathbf{R_{2\times3}}$ as follows. Matrices whose rows and columns are arranged in descending order will be said to be in \it standard form\rm. It is straightforward to see that only five of the~60 matrices we have just constructed have this form, namely
\begin{equation}
\label{minz}
\left( \begin{array}{ccc}a & b & c \\d & e & f\end{array}\right),
\left( \begin{array}{ccc}a & b & d \\c & e & f\end{array}\right),
\left( \begin{array}{ccc}a & b & e \\c & d & f\end{array}\right),
\left( \begin{array}{ccc}a & c & d \\b & e & f\end{array}\right),\textrm{ and }
\left( \begin{array}{ccc}a & c & e \\b & d & f\end{array}\right).
\end{equation}
Notice that all of these are in the form (\ref{mates}) with the additional condition that~$u>v>w$.
If we allow the bottom row of any of these to be permuted we obtain~$5=|\mathbf{S}_3|-1$ new matrices which are not in standard form. In all this gives a total of~$30$ matrices split into five groups of~$6$, indexed by each matrix in~(\ref{minz}).

Now consider matrices in~$\mathbf{R_{2\times3}}$ which cannot be in standard form by virtue of having top row entries which are ``too small'' but nevertheless which still have the rows in descending order, viz:
\begin{equation}
\label{minzoneup}
\left( \begin{array}{ccc}a & b & f \\c & d & e\end{array}\right),
\left( \begin{array}{ccc}a & c & f \\b & d & e\end{array}\right),
\left( \begin{array}{ccc}a & d & e \\b & c & f\end{array}\right),
\left( \begin{array}{ccc}a & d & f \\b & c & e\end{array}\right),\textrm{ and }
\left( \begin{array}{ccc}a & e & f \\b & c & d\end{array}\right).
\end{equation}
Once again, by permuting the bottom row of each we obtain five new matrices: a second total of $30$ matrices split into five groups of $6$, indexed by each matrix in (\ref{minzoneup}).

To visualize these subsets of matrices see the diagram on page~\pageref{hexcomb} (together with the classification of $\mathbf{R_{2\times3}}$ set out in appendix~\ref{frog} which is the key to their enumeration).

In order to facilitate the proof of theorem~\ref{biggun}, here are a few results which help us to classify the relations between the~$\mathbf{R_{2\times3}}$ classes. Call two matrices $M = \left( \begin{array}{ccc}a & p & q\\r & s & t\end{array}\right),\ \ N = \left( \begin{array}{ccc}a & x & y\\u & v & w\end{array}\right)$ \it lexicographically ordered \rm if the pair of row vectors $(a,p,q,r,s,t)$ and $(a,x,y,u,v,w)$ is so ordered (ie the word ``apqrst'' would precede the word ``axyuvw'' in an English dictionary).

\begin{lemma}
We may order the matrices in $\mathbf{R_{2\times3}}$ lexicographically, and majorisation respects that ordering.
\end{lemma}

That is to say, if $M$ lies above $N$ lexicographically then $N$ \bf cannot \rm majorise $M$. Note that this is \bf not \rm the case for the relation $\rhd$. 

\begin{proof}
The existence of such an ordering is obvious; so let $M\succ N$ where $M,N\in\mathbf{R_{2\times3}}$. We need to show that $M$ precedes $N$ in the lexicographical ordering. Looking first at the top row of each matrix: since they both contain $a$, and since a sum containing $a$ can only be \it a priori \rm majorised by another sum containing $a$, it follows that the top row sum of $M$ must be $\geq$ the top row sum of $N$ \it a priori\rm. But both top rows are ordered lexicographically. So the top row of $M$ either precedes that of $N$, in which case we are done; or else the top rows are in fact equal and so we must look at the columns. But this is just the argument of lemma~\ref{treecritter}: see below.
\end{proof}

\begin{lemma}\label{treecritter}
Fix any matrix $M=\left( \begin{array}{ccc}a & x & y\\u & v & w\end{array}\right) \in\mathbf{R_{2\times3}}$ with the additional requirement that $u>v>w$. Permuting the elements of the bottom row under the action of the symmetric group $\mathbf{S}_3$ we have the following majorisation relations:
\begin{equation}
\begin{array}{ccccccc}\label{flea}
\left( \begin{array}{ccc}a & x & y\\u & v & w\end{array}\right)                                 & \succ &
\begin{array}{c}\left( \begin{array}{ccc}a & x & y\\u & w & v\end{array}\right)\\
\left( \begin{array}{ccc}a & x & y\\v & u & w \end{array}\right) \end{array}                     & \succ &
\begin{array}{c}\left( \begin{array}{ccc}a & x & y\\v & w & u \end{array}\right)\\
\left( \begin{array}{ccc}a & x & y\\w & u & v\end{array}\right) \end{array}                     & \succ &
\left( \begin{array}{ccc}a & x & y\\w & v & u\end{array}\right)
\end{array}\end{equation}
There are in general no \rm a priori \it majorisation relations within the two vertical pairs.
\end{lemma}

\begin{proof}
Noting that only two of the column sums are changed at each step, apply definition~\ref{critta} bearing in mind the assumptions that~$u>v>w$ and $a>x>y$.
\end{proof}

Now the rightmost matrix in~(\ref{flea}) corresponds to left multiplication by the permutation $\varpi=(m_{21},m_{23})$ of the matrix $M=(m_{ij})$, that is:  
$\left( \begin{array}{ccc}a & x & y\\w & v & u\end{array}\right) = \varpi\left( \begin{array}{ccc}a & x & y\\u & v & w\end{array}\right).$
By proposition~\ref{majmcmaj} the fact that $A$ majorises $B$ implies that $I(A)<I(B)$, so the minimal value for the CMI among the \emph{representative} matrices in $\Rtt$ must occur in a matrix of the form on the left-hand side of~(\ref{flea}); conversely the maximum must occur in a matrix of the form on the right-hand side of~(\ref{flea}).

\begin{corollary}\label{bighitter}
There is some $M$ in (\ref{minz}) such that the \bf minimal\it\ value for the CMI of any matrix from the set $\mathbf{R_{2\times3}}$ is given by~$I(M)$.

There is some $A$ in (\ref{minzoneup}) such that the \bf maximal\it\ value for the CMI of any matrix from the set $\mathbf{R_{2\times3}}$ is given by~$I(\varpi(A))$.
\end{corollary}

\begin{proof}
Recall from lemma~\ref{majoris} the relationship between matrices, majorisation and the CMI.
For the minima, one can show directly from the definitions that every matrix in~(\ref{minzoneup}) is majorised by some matrix in~(\ref{minz}), and we then use lemma~\ref{treecritter}.

For the maxima we again use lemma~\ref{treecritter} to reduce the problem to comparing the matrices obtained by applying~$\varpi$ to~(\ref{minz}), with those obtained by applying~$\varpi$ to~(\ref{minzoneup}): then it can be shown once again from the definitions that every matrix in~(\ref{minz}) majorises at least one matrix in~(\ref{minzoneup}). Indeed it may be shown directly that each matrix in~(\ref{minz}) majorises the matrix~$X$ in~(\ref{exxon}).
\end{proof}

\subsubsection*{Aside: the basic majorisation structure in pictures}

Using the simple majorisation relations developed in the foregoing discussion we have established a kind of ``honeycomb'' which is the backbone of the partial order which is elaborated upon in~\cite{mcmaj}. Figure~\ref{hexcomb} shows the basic hexagonal frames corresponding to the majorisation orderings in~(\ref{flea}). The honeycomb consists of~10 hexagons each containing~6 matrices (one row of~5 slightly below the other reflecting the ``standard'' classification), with each matrix linked via a hexagonal pattern to the other matrices in its own group.
Each hexagonal cell is in itself a diagram of the Bruhat order on $\mathbf{S}_3$. The 2 sets of 5 hexagons come from lemma~\ref{treecritter}; and the 12 lines of 5 matrices each (consisting of aligned vertices of the hexagons in their respective groupings) arise from variants of~(\ref{minz}) and~(\ref{minzoneup}). The red numbers represent the ``major'' element in each hexagon and are in fact all of the matrices in~(\ref{minz}) for the top row, and~(\ref{minzoneup}) for the bottom row. Note that we have placed the maximal CMI element $\mathbf{48}$ at the very bottom point, reflecting the fact that it is below every other matrix in the partial order induced by~$\rhd$. The minimal CMI will occur for a matrix on the very top row (matrices $\mathbf{1},\mathbf{7},\mathbf{13},\mathbf{25}$ or $\mathbf{31}$).

The numbering is as per appendix~\ref{frog}, ie the lexicographic ordering.
We have stuck to this ordering as much as possible in the diagram itself, trying to increase numbers within the hexagons as we move down and from left to right; however in places we have changed it slightly so that the patterns are rendered more clearly. The black arrows represent the majorisation relations in lemma~\ref{treecritter} which arise within each hexagon. 

The light blue double-headed arrows represent the action of the inner automorphism $\xi_\omega$ arising from the unique element $\omega=(1,6)(2,5)(3,4)\in\Sym_6$ of maximal length (viewing $\Sym_6$ as a Coxeter group - see~\cite{bjorner} for definitions) which flips 22 pairs of matrix classes and fixes the remaining 16. Since this automorphism respects the relation $\rhd$ on $\Rtt$ it follows that any entropic relations (including of course majorisation) involving the nodes which have a blue arrow pointing to them will occur in pairs, thus considerably simplifying the structure. We explore this in~\cite{mcmaj}.

\begin{figure}[h!btp]
\begin{center}
\mbox{
\subfigure{\includegraphics[width=6in,height=6in,keepaspectratio]{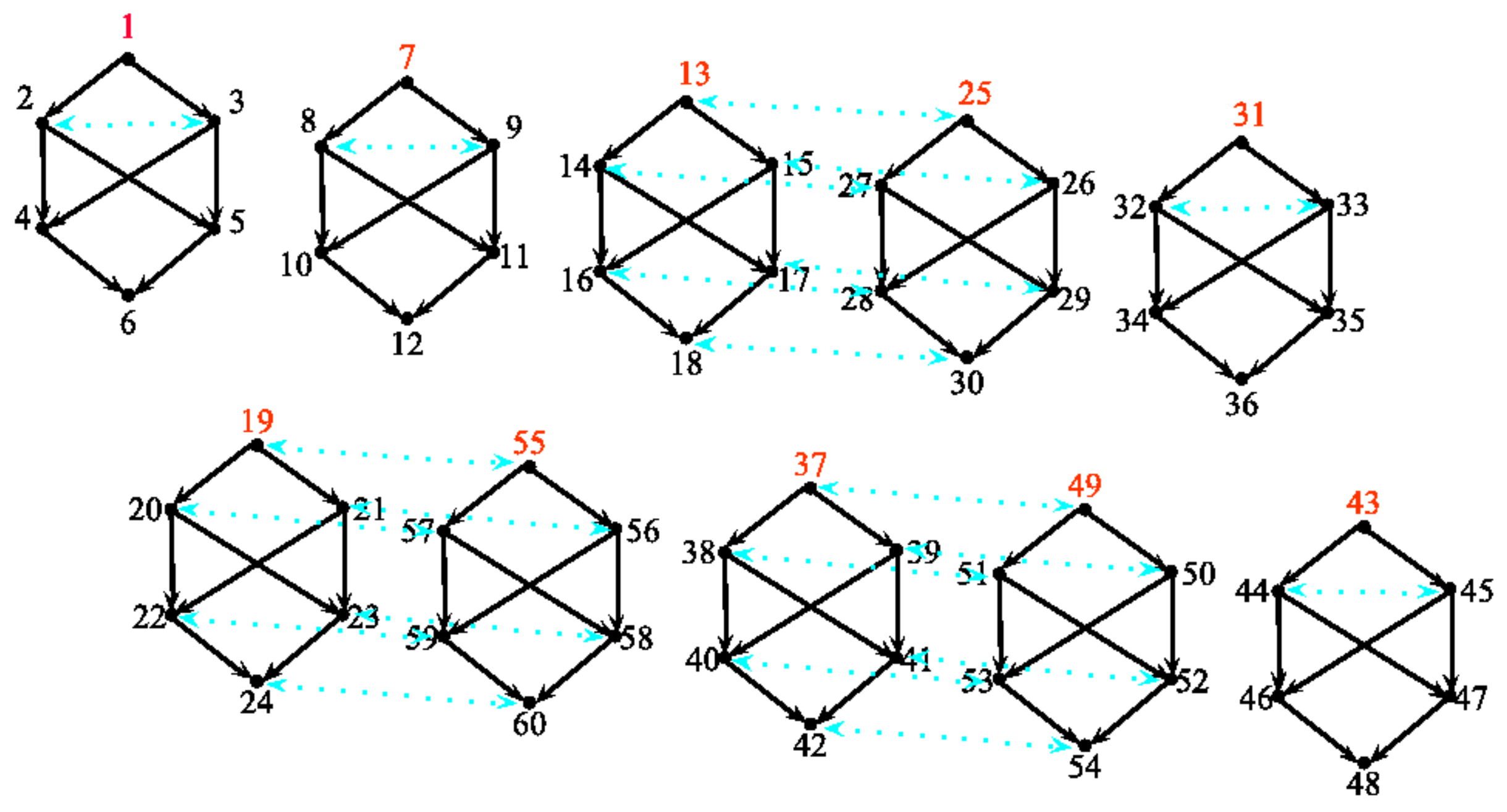}}
}
\end{center}\caption{Representation of the most basic horizontal-transposition-based majorisation relations on $\Rtt$ together with the action of the involution $\xi$ (the blue arrows)}
\label{hexcomb}
\end{figure}

\subsubsection{Completion of the proof of theorem~\ref{biggun}}

We are reduced by corollary~\ref{bighitter} to showing that the CMI of the matrix $X$ is greater than or equal to that of the other $4$ matrices obtained by applying $\varpi$ to the remainder of (\ref{minzoneup}). We state these matrices for convenience:
\begin{equation}\label{minzoneupmod}
Y_1 = \left( \begin{array}{ccc}a & b & f \\e & d & c\end{array}\right),\ \
Y_2 = \left( \begin{array}{ccc}a & c & f \\e & d & b\end{array}\right),\ \
Y_3 = \left( \begin{array}{ccc}a & d & f \\e & c & b\end{array}\right),\ \
Y_4 = \left( \begin{array}{ccc}a & e & f \\d & c & b\end{array}\right).
\end{equation}
We should remark first that there is no \it a priori \rm majorisation relationship between any of these matrices $X,Y_1,Y_2,Y_3,Y_4$. So we need a weaker (easier to satisfy) condition which distinguishes between them, namely the entropic relation~$\rhd$.

Since $P\rhd Q$ implies that the CMI of $P$ is lower than that of $Q$, it is a transitive relation on CMI of matrices and so it will suffice to show the following relations:
\begin{eqnarray*}
(i)\ Y_1 & \rhd & Y_2 \\
(ii)\ Y_2 & \rhd & X  \\
(iii)\ Y_4 & \rhd & Y_3 \\
(iv)\ Y_3 & \rhd & X,
\end{eqnarray*}

For all of these we just apply proposition~\ref{titrate} as follows.

(i) Let $P=Y_1$, $\alpha=b$, $\beta=c$ and $\tau$ swaps entries $p_{12}$ and $p_{23}$ (ie $\alpha$ and $\beta$), yielding $P^\tau = Y_2$. Using the definitions we see that $r_\alpha^\tau=a+c+f$, $c_\alpha^\tau=c+d$, $r_\beta=c+d+e$ and $c_\beta=c+f$. So indeed $c_\beta$ is always the minimal value of these four, and we see that in addition the hypothesis that $r_\beta+c_\beta \leq r_\alpha^\tau+c_\alpha^\tau$ is satisfied. So $P\rhd P^\tau$ by proposition~\ref{titrate}.

(iii) This time let $P=Y_4$, $\alpha=d$, $\beta=e$ and $\tau$ swaps entries $p_{12}$ and $p_{21}$, yielding $P^\tau = Y_3$. Again using the definitions: $r_\alpha^\tau=b+c+e$, $c_\alpha^\tau=a+e$, $r_\beta=a+e+f$ and $c_\beta=c+e$. So once again $c_\beta$ is always the minimal value of these four, and the hypothesis that $r_\beta+c_\beta \leq r_\alpha^\tau+c_\alpha^\tau$ is again satisfied. So $P\rhd P^\tau$ by proposition~\ref{titrate}.

(iv) Now $P=Y_3$, $\alpha=e$, $\beta=f$ and $\tau$ swaps entries $p_{13}$ and $p_{21}$, yielding $P^\tau = X$.
We have $r_\alpha^\tau=b+c+f$, $c_\alpha^\tau=a+f$, $r_\beta=a+d+f$ and $c_\beta=b+f$. Again $c_\beta$ is always the minimal value of these four, and the hypothesis that $r_\beta+c_\beta \leq r_\alpha^\tau+c_\alpha^\tau$ is satisfied. So $P\rhd P^\tau$ by proposition~\ref{titrate}.

(ii) Finally the slightly trickier case of proving $P=Y_2 \rhd X$. The reason this is different is that on the face of it, it does not consist of a single transposition but rather of a product of two disjoint transpositions $(p_{12},p_{22})(p_{13},p_{21})$ for which the intermediate matrices have no \it a priori \rm relations with one another. However we may use the $\rhd$-relational framework above if we observe that in fact the single transposition $\tau=(p_{11},p_{23})$ with $\alpha=a$ and $\beta=b$ gives $P^\tau=\left( \begin{array}{ccc}b & c & f \\e & d & a\end{array}\right)$ which is seen to have CMI equal to that of $X$. Since there is nothing in the definition of~$\rhd$ which requires a matrix to be in $\mathbf{R_{2\times3}}$ (this is merely a convenient classification for keeping track of them), we may apply the same techniques as in (i), (iii) and (iv) to conclude that $r_\alpha^\tau=b+c+f$, $c_\alpha^\tau=b+e$, $r_\beta=b+d+e$ and $c_\beta=b+f$: so $c_\beta$ is always the minimal value
and since $b+f+b+d+e<b+e+b+c+f$ we are again able to apply proposition~\ref{titrate} to conclude that the quantity $I(X)-I(Y_2)=I(P^\tau)-I(P)$ is positive, as required.

This completes the proof of theorem~\ref{biggun}.\qed

\begin{remark}\label{ream}
It is worth pointing out that one may arrive at the conclusion of theorem~\ref{biggun} by a process of heuristic reasoning, as follows. Recall from definition~\ref{cmile} that the CMI consists of three components, of which the last one is identical for all matrices which are permutations of one another. So in order to understand maxima/minima we restrict our focus to the first two terms, namely the entropies of the marginal probability vectors. Now entropy is a measure of the ``randomness'' of the marginal probabilities: the more uniform they are the higher will be the contribution to the CMI from these row and column sum vectors. Beginning with the columns, if we look at the a priori ordering $a>b>c>d>e>f$ it is evident that the most uniform way of selecting pairs in general so as to be as close as possible to one another would be to begin at the outside and work our way in: namely the column sum vector should read $(a+f,\ b+e,\ c+d)$. Similarly for the row sums: we need to add small terms to $a$, but the position of $f$ is already taken in the same column as $a$, so that just leaves $d$ and $e$ in the top row, and $c$ and $b$ fill up the bottom row in the order dictated by the column sums. We perform a similar analysis for the simpler case of 2x2 matrices in appendix~\ref{2by2}, where in fact we can achieve a total ordering by the same method.
\end{remark}

\section*{Acknowledgements}
We would like to thank Terry Rudolph and the QOLS department of Imperial College London for their generous hospitality for the duration of this work. DJ~is supported by the Royal Commission for the Exhibition of 1851.

\newpage
\appendix
\section{The matrix class representatives in $\mathbf{R_{2\times3}}$}\label{frog}

We list the matrix representatives in~$\mathbf{R_{2\times3}}$ in lexicographic order together with the enumeration we have used throughout the paper when referring to them, alongside in each case the element~$\sigma\in G=\Ssix$ in cycle notation which represents the appropriate permutation of the fiducial matrix~$\left( \begin{array}{ccc}a & b & c \\d & e & f\end{array}\right)$ which we have chosen to represent the identity~$()\in G$. Note that each~$\sigma$ is only chosen up to row- and column-swaps. Also, since we have chosen to represent the matrices with~$a$ in the top left-hand corner and with decreasing top row, the set of representative cycles displayed is effectively a copy of~$\Sfive$ modulo a subgroup of order~2.\\

\begin{tabular}{| l l l | l l l |l l l | l l l |}
\hline
 $\mathbf{1}:$ & $\left( \begin{array}{ccc}a & b & c\\d & e & f\end{array}\right)$, & $\mathbf{()}$ &
 $\mathbf{2}:$ & $\left( \begin{array}{ccc}a & b & c\\d & f & e\end{array}\right)$, & $\mathbf{(56)}$ &
 $\mathbf{3}:$ & $\left( \begin{array}{ccc}a & b & c\\e & d & f\end{array}\right)$, & $\mathbf{(45)}$ & 
 $\mathbf{4}:$ & $\left( \begin{array}{ccc}a & b & c\\e & f & d\end{array}\right)$, & $\mathbf{(465)}$ \\ \hline 
 $\mathbf{5}:$ & $\left( \begin{array}{ccc}a & b & c\\f & d & e\end{array}\right)$, & $\mathbf{(456)}$ & 
 $\mathbf{6}:$ & $\left( \begin{array}{ccc}a & b & c\\f & e & d\end{array}\right)$, & $\mathbf{(46)}$ &
 $\mathbf{7}:$ & $\left( \begin{array}{ccc}a & b & d\\c & e & f\end{array}\right)$, & $\mathbf{(34)}$ & 
 $\mathbf{8}:$ & $\left( \begin{array}{ccc}a & b & d\\c & f & e\end{array}\right)$, & $\mathbf{(34)(56)}$ \\ \hline  
 $\mathbf{9}:$ & $\left( \begin{array}{ccc}a & b & d\\e & c & f\end{array}\right)$, & $\mathbf{(354)}$ & 
 $\mathbf{10}:$ & $\left( \begin{array}{ccc}a & b & d\\e & f & c\end{array}\right)$, & $\mathbf{(3654)}$ & 
 $\mathbf{11}:$ & $\left( \begin{array}{ccc}a & b & d\\f & c & e\end{array}\right)$, & $\mathbf{(3564)}$ & 
 $\mathbf{12}:$ & $\left( \begin{array}{ccc}a & b & d\\f & e & c\end{array}\right)$, & $\mathbf{(364)}$ \\ \hline
 $\mathbf{13}:$ & $\left( \begin{array}{ccc}a & b & e\\c & d & f\end{array}\right)$, & $\mathbf{(345)}$ & 
 $\mathbf{14}:$ & $\left( \begin{array}{ccc}a & b & e\\c & f & d\end{array}\right)$, & $\mathbf{(3465)}$ & 
 $\mathbf{15}:$ & $\left( \begin{array}{ccc}a & b & e\\d & c & f\end{array}\right)$, & $\mathbf{(35)}$ & 
 $\mathbf{16}:$ & $\left( \begin{array}{ccc}a & b & e\\d & f & c\end{array}\right)$, & $\mathbf{(365)}$ \\ \hline 
 $\mathbf{17}:$ & $\left( \begin{array}{ccc}a & b & e\\f & c & d\end{array}\right)$, & $\mathbf{(35)(46)}$ & 
 $\mathbf{18}:$ & $\left( \begin{array}{ccc}a & b & e\\f & d & c\end{array}\right)$, & $\mathbf{(3645)}$ &
 $\mathbf{19}:$ & $\left( \begin{array}{ccc}a & b & f\\c & d & e\end{array}\right)$, & $\mathbf{(3456)}$ & 
 $\mathbf{20}:$ & $\left( \begin{array}{ccc}a & b & f\\c & e & d\end{array}\right)$, & $\mathbf{(346)}$ \\ \hline 
 $\mathbf{21}:$ & $\left( \begin{array}{ccc}a & b & f\\d & c & e\end{array}\right)$, & $\mathbf{(356)}$ & 
 $\mathbf{22}:$ & $\left( \begin{array}{ccc}a & b & f\\d & e & c\end{array}\right)$, & $\mathbf{(36)}$ & 
 $\mathbf{23}:$ & $\left( \begin{array}{ccc}a & b & f\\e & c & d\end{array}\right)$, & $\mathbf{(3546)}$ & 
 $\mathbf{24}:$ & $\left( \begin{array}{ccc}a & b & f\\e & d & c\end{array}\right)$, & $\mathbf{(36)(45)}$ \\ \hline 
 $\mathbf{25}:$ & $\left( \begin{array}{ccc}a & c & d\\b & e & f\end{array}\right)$, & $\mathbf{(243)}$ & 
 $\mathbf{26}:$ & $\left( \begin{array}{ccc}a & c & d\\b & f & e\end{array}\right)$, & $\mathbf{(243)(56)}$ & 
 $\mathbf{27}:$ & $\left( \begin{array}{ccc}a & c & d\\e & b & f\end{array}\right)$, & $\mathbf{(2543)}$ & 
 $\mathbf{28}:$ & $\left( \begin{array}{ccc}a & c & d\\e & f & b\end{array}\right)$, & $\mathbf{(26543)}$  \\ \hline 
 $\mathbf{29}:$ & $\left( \begin{array}{ccc}a & c & d\\f & b & e\end{array}\right)$, & $\mathbf{(25643)}$ & 
 $\mathbf{30}:$ & $\left( \begin{array}{ccc}a & c & d\\f & e &b\end{array}\right)$, & $\mathbf{(2643)}$ &
 $\mathbf{31}:$ & $\left( \begin{array}{ccc}a & c & e\\b & d & f\end{array}\right)$, & $\mathbf{(2453)}$ & 
 $\mathbf{32}:$ & $\left( \begin{array}{ccc}a & c & e\\b & f & d\end{array}\right)$, & $\mathbf{(24653)}$ \\ \hline 
 $\mathbf{33}:$ & $\left( \begin{array}{ccc}a & c & e\\d & b & f\end{array}\right)$, & $\mathbf{(253)}$ & 
 $\mathbf{34}:$ & $\left( \begin{array}{ccc}a & c & e\\d & f & b\end{array}\right)$, & $\mathbf{(2653)}$ & 
 $\mathbf{35}:$ & $\left( \begin{array}{ccc}a & c & e\\f & b & d\end{array}\right)$, & $\mathbf{(253)(46)}$ & 
 $\mathbf{36}:$ & $\left( \begin{array}{ccc}a & c & e\\f & d &b\end{array}\right)$, & $\mathbf{(26453)}$ \\ \hline
 $\mathbf{37}:$ & $\left( \begin{array}{ccc}a & c & f\\b & d & e\end{array}\right)$, & $\mathbf{(24563)}$ & 
 $\mathbf{38}:$ & $\left( \begin{array}{ccc}a & c & f\\b & e & d\end{array}\right)$, & $\mathbf{(2463)}$ & 
 $\mathbf{39}:$ & $\left( \begin{array}{ccc}a & c & f\\d & b & e\end{array}\right)$, & $\mathbf{(2563)}$ & 
 $\mathbf{40}:$ & $\left( \begin{array}{ccc}a & c & f\\d & e & b\end{array}\right)$, & $\mathbf{(263)}$  \\ \hline 
 $\mathbf{41}:$ & $\left( \begin{array}{ccc}a & c & f\\e & b & d\end{array}\right)$, & $\mathbf{(25463)}$ & 
 $\mathbf{42}:$ & $\left( \begin{array}{ccc}a & c & f\\e & d &b\end{array}\right)$, & $\mathbf{(263)(45)}$ & 
 $\mathbf{43}:$ & $\left( \begin{array}{ccc}a & d & e\\b & c & f\end{array}\right)$, & $\mathbf{(24)(35)}$ & 
 $\mathbf{44}:$ & $\left( \begin{array}{ccc}a & d & e\\b & f & c\end{array}\right)$, & $\mathbf{(24)(365)}$ \\ \hline 
 $\mathbf{45}:$ & $\left( \begin{array}{ccc}a & d & e\\c & b & f\end{array}\right)$, & $\mathbf{(2534)}$ & 
 $\mathbf{46}:$ & $\left( \begin{array}{ccc}a & d & e\\c & f & b\end{array}\right)$, & $\mathbf{(26534)}$ & 
 $\mathbf{47}:$ & $\left( \begin{array}{ccc}a & d & e\\f & b & c\end{array}\right)$, & $\mathbf{(25364)}$ & 
 $\mathbf{48}:$ & $\left( \begin{array}{ccc}a & d & e\\f & c &b\end{array}\right)$, & $\mathbf{(264)(35)}$ \\ \hline 
 $\mathbf{49}:$ & $\left( \begin{array}{ccc}a & d & f\\b & c & e\end{array}\right)$, & $\mathbf{(24)(356)}$ & 
 $\mathbf{50}:$ & $\left( \begin{array}{ccc}a & d & f\\b & e & c\end{array}\right)$, & $\mathbf{(24)(36)}$ & 
 $\mathbf{51}:$ & $\left( \begin{array}{ccc}a & d & f\\c & b & e\end{array}\right)$, & $\mathbf{(25634)}$ & 
 $\mathbf{52}:$ & $\left( \begin{array}{ccc}a & d & f\\c & e & b\end{array}\right)$, & $\mathbf{(2634)}$  \\ \hline 
 $\mathbf{53}:$ & $\left( \begin{array}{ccc}a & d & f\\e & b & c\end{array}\right)$, & $\mathbf{(254)(36)}$ & 
 $\mathbf{54}:$ & $\left( \begin{array}{ccc}a & d & f\\e & c &b\end{array}\right)$, & $\mathbf{(26354)}$ & 
 $\mathbf{55}:$ & $\left( \begin{array}{ccc}a & e & f\\b & c & d\end{array}\right)$, & $\mathbf{(24635)}$ & 
 $\mathbf{56}:$ & $\left( \begin{array}{ccc}a & e & f\\b & d & c\end{array}\right)$, & $\mathbf{(245)(36)}$  \\ \hline 
 $\mathbf{57}:$ & $\left( \begin{array}{ccc}a & e & f\\c & b & d\end{array}\right)$, & $\mathbf{(25)(346)}$ & 
 $\mathbf{58}:$ & $\left( \begin{array}{ccc}a & e & f\\c & d & b\end{array}\right)$, & $\mathbf{(26345)}$ & 
 $\mathbf{59}:$ & $\left( \begin{array}{ccc}a & e & f\\d & b & c\end{array}\right)$, & $\mathbf{(25)(36)}$ & 
 $\mathbf{60}:$ & $\left( \begin{array}{ccc}a & e & f\\d & c &b\end{array}\right)$, & $\mathbf{(2635)}$ \\ \hline 
\end{tabular}

\section{The $2\times2$ case}\label{2by2}

We set out here a detailed proof of the phenomenon of maximal and minimal CMI in the case of $2\times2$-matrices, which was first proven in~\cite{santerdav} and~\cite{santerdavPRL}. We adopt a quite different approach, more direct in some sense than going via majorisation theory, because it gives some insight into what is really going on in the $2\times3$-case.

Let $a>b>c>d>0$ with $a+b+c+d=1$ and let
$M = \left( \begin{array}{cc}a & b\\ c& d\end{array}\right)$
be the corresponding $2\times2$ probability matrix. 
Define the (classical) mutual information $I(M)$ as before (though in a slightly different but equivalent form to definition~\ref{cmile}) to be
$$I(M) = h(a+b)+h(a+c)-\sum_{x=a,b,c,d}-x\log x,$$
where $h(x)$ is the standard binary entropy function
$$h(x) = -x\log x - (1-x)\log(1-x) = H(x)+H(1-x),$$
defined for $ x\in [0,1]$. Again, we wish to establish whether there is a permutation of the elements of $M$ which gives us \it a priori \rm the minimal or maximal possible mutual information.

In analogy with the $2\times3$ situation above, we may consider the action of~$\mathbf{S}_4$ on the matrix~$M$, denoting the places of the matrix by~$\left( \begin{array}{cc}1 & 2\\ 3& 4\end{array}\right)$. Denote by $\cal{M}_{\rm 2\times2}$ the space of all 24 possible permutations of the matrix $M$ (for a fixed choice of $a,b,c,d$). We may observe once again that the CMI of $M$ is invariant under a large subgroup $J$ of $\mathbf{S}_4$. Namely it is unchanged if we swap the rows, swap the columns or transpose the matrix. Given our choice of numbering the generators of our subgroup of $\mathbf{S}_4$ are canonically the elements $(1,3)(2,4)$, $(1,2)(3,4)$ and $(2,3)$ respectively. The row and column swap operations commute with one another, but the transpose operation $(2,3)$ causes our subgroup $J$ to be a copy of the dihedral group $D_8$. Explicitly:
$$J=\{(),(1,2,4,3),(1,4)(2,3),(1,3,4,2),(1,2)(3,4),(1,3)(2,4),(1,4),(2,3)\},$$
and so we may write the right coset space as
$$\mathbf{S}_4/J=\{()J,(2,4)J,(3,4)J\},$$
viewing the action of $\mathbf{S}_4$ as being via right translation, which in turn 
corresponds to a set of matrix representatives of each class respectively as
$$M=\left( \begin{array}{cc}a & b\\ c& d\end{array}\right),\ M^{(2,4)}=\left( \begin{array}{cc}a & d\\ c& b\end{array}\right)\ {\rm and}\ M^{(3,4)}=\left( \begin{array}{cc}a & b\\ d& c\end{array}\right).$$

In the $2\times3$ case we were able to find a unique maximum and 5 possible minima. In this much simpler case we can in fact order all three right coset classes \it a priori.\rm

\begin{proposition}\label{2free}
With notation as above, 
$$I(M)<I(M^{(3,4)})<I(M^{(2,4)}).$$
\end{proposition}

\begin{remark}
Recall the remark on page~\pageref{ream} after the proof of theorem~\ref{biggun}: it is possible to arrive at the conclusion of proposition~\ref{2free} by heuristic reasoning as follows. Following the method there we focus solely on the entropies of the marginal probability vectors (the row and column sum vectors in the text). Uniformity in these will yield higher entropies, hence for the maximum we should seek to have the row and column sum vectors each as near to $(0.5,0.5)$ as possible. Clearly this will occur in general when we add $a$ to $d$ and $b$ to $c$; however this cannot occur for both rows and sums so the next best thing is to have $a+c$ and $b+d$. Hence the maximal CMI will occur for the matrix~$\left( \begin{array}{cc}a & d\\ c& b\end{array}\right)$. By similar reasoning the minimum must occur for the least uniform sums, namely $a+b$ with $c+d$ and then $a+c$ and $b+d$ leading to the minimal CMI occurring for~$\left( \begin{array}{cc}a & b\\ c& d\end{array}\right)$. This leaves the middle value for the remaining matrix, which has of course the maximum-entropy set $(a+d,\ b+c)$ together with the minimum-entropy set $(a+b,\ c+d)$.
\end{remark}

\begin{proof}
From the discussion above it follows that the function $I$ takes on at most three distinct values on the orbit of $M$ under the action of $\mathbf{S_4}$. These three values are $I(M)$, $I(M^{(2,4)})$ and $I(M^{(3,4)})$. So we have in the notation introduced above,
$$I(M^{(2,4)}) = h(a+d)+h(a+c)-\sum_{x=a,b,c,d}-x\log x,$$
and
$$I(M^{(3,4)}) = h(a+b)+h(a+d)-\sum_{x=a,b,c,d}-x\log x.$$
Hence in order to prove the proposition we may simply consider the differences
$$I(M^{(3,4)})-I(M) = h(a+d)-h(a+c)$$
and
$$I(M^{(2,4)})-I(M^{(3,4)}) = h(a+c)-h(a+b).$$
The claim of the proposition is that both of these quantities are non-negative.

Here we need three basic properties of the function $h(x)$ on its domain of definition (the unit interval), namely that it is continuous, symmetric about the line $x=\frac{1}{2}$ and monotonic decreasing on either side of that line (moving always in the direction away from the central point of course). Note that this is weaker than needing concavity and maxima from calculus.

Given these three conditions, the size of $h(x)$ versus $h(y)$ for $x,y\in[0,1]$ is measured precisely by how close each of $x$ and $y$ is to the point $x=\frac{1}{2}$. In other words, if $|x-\frac{1}{2}|<|y-\frac{1}{2}|$ then $h(x)>h(y)$. Hence we are reduced to showing that
\begin{equation}\label{smd}
|a+d-\frac{1}{2}|\ <\ |a+c-\frac{1}{2}|\ <\ |a+b-\frac{1}{2}|.
\end{equation}
By the ordering $a>b>c>d$ and the fact that $a+b+c+d=1$, both
$$a+b>\frac{1}{2},\ a+c>\frac{1}{2}.$$

Notice that $a+d$ may be either side of $\frac{1}{2}$; however
\begin{equation}\label{glomp}
|a+d-\frac{1}{2}|=\frac{1}{2}|2a+2d-1|=\frac{1}{2}|a+d-b-c|\leq|\frac{a-c}{2}|+|\frac{b-d}{2}|=\frac{a-c}{2}+\frac{b-d}{2}
\end{equation}
by the triangle inequality, noting for the last equality that $a\geq c$ and $b\geq d$ by assumption. By symmetry we may also write this as:
\begin{equation}\label{glump}
|a+d-\frac{1}{2}|\leq\frac{a-b}{2}+\frac{c-d}{2}.
\end{equation}

Now using the facts about $a,b,c,d$ once again we have:
\begin{equation}\label{teeny}
|a+c-\frac{1}{2}|=a+c-\frac{1}{2}=\frac{1}{2}(2a+2c-1)=\frac{1}{2}(a+c-b-d)=\frac{a-b}{2}+\frac{c-d}{2}
\end{equation}
and
\begin{equation}\label{weeny}
|a+b-\frac{1}{2}|=a+b-\frac{1}{2}=\frac{1}{2}(2a+2b-1)=\frac{1}{2}(a+b-c-d)=\frac{a-c}{2}+\frac{b-d}{2},
\end{equation}
which combined with~(\ref{glomp}) and~(\ref{glump}) prove directly that $|a+d-\frac{1}{2}|\ <\ |a+c-\frac{1}{2}|$ and that $|a+d-\frac{1}{2}|\ <\ |a+b-\frac{1}{2}|$. But it is clear moreover from the fact that $a>b>c>d$ that (\ref{teeny})$<$(\ref{weeny}), which proves the rest of the inequality in~(\ref{smd}).
\end{proof}

\end{document}